\documentclass[reqno]{amsart}
\usepackage{graphicx} 
\usepackage{tikz-cd}
\usepackage{enumerate}
\usepackage{bbm}
\usepackage{color}
\usepackage{amsaddr}
\usepackage{amsthm}
\usepackage{bm}
\usepackage{mathtools}
\usepackage{enumitem}
\usepackage{etoolbox}
\usepackage[
        doi=false,
	backend=biber,
	style=alphabetic,
	]{biblatex}
\NewBibliographyString{bywithappendix}
\DefineBibliographyStrings{english}{
  bywithappendix = {Appendix to},
}
\DeclareFieldFormat{title}{\mkbibquote{#1}}
\usepackage{hyperref}
\hypersetup{
    colorlinks=true,
    linkcolor=blue,
    citecolor=blue,
    pdfpagemode=FullScreen,
    }
\usepackage[capitalise,noabbrev,nameinlink]{cleveref}
\makeatletter
\renewcommand{\eqref}[1]{%
  \hyperref[#1]{\textup{\tagform@{\ref*{#1}}}}%
}
\makeatother
\makeatletter
\newcommand{\customlabel}[2]{%
   \protected@write \@auxout {}{\string \newlabel {#1}{{#2}{\thepage}{#2}{#1}{}} }%
   \hypertarget{#1}{}
}
\makeatother
\title[RKHS methods for modelling the discount curve]{Reproducing kernel Hilbert space methods for modelling the discount curve}

\author{Andreas Celary$^{\dag}$, Paul Kr\"uhner$^{\dag}$, Zehra Eksi$^{\dag}$}
\address{$^\dag$Institute for Statistics and Mathematics, WU-University of Economics and Business} 
\email{acelary@wu.ac.at, peisenbe@wu.ac.at, zehra.eksi-altay@wu.ac.at}
\newtheorem{theorem}{Theorem}[section]
\newtheorem{definition}[theorem]{Definition}
\newtheorem{lemma}[theorem]{Lemma}

\newtheorem{proposition}[theorem]{Proposition}

\newtheorem{corollary}[theorem]{Corollary}
\newtheorem{remark}[theorem]{Remark}
\newtheorem{example}[theorem]{Example}

\Crefname{theorem}{Theorem}{Theorems}
\Crefname{theoremenumi}{Theorem}{Theorems}
\AtBeginEnvironment{theorem}{%
    \crefalias{enumi}{theoremenumi}%
    \setlist[enumerate,1]{
        label={\textit{\roman*)}},
        ref={\thetheorem.\roman*)}
    }%
}
\graphicspath{{./figures}}
\Crefname{lemma}{Lemma}{Lemmas}
\Crefname{lemmaenumi}{Lemma}{Lemmas}
\AtBeginEnvironment{lemma}{%
    \crefalias{enumi}{lemmmaenumi}%
    \setlist[enumerate,1]{
        label={\textit{\roman*)}},
        ref={\thelemma.\roman*)}
    }%
}

\Crefname{definition}{Definition}{Definitions}
\Crefname{definitionenumi}{Definition}{Definitions}
\AtBeginEnvironment{definition}{%
    \crefalias{enumi}{definitionenumi}%
    \setlist[enumerate,1]{
        label={\textit{\roman*)}},
        ref={\thedefinition.\roman*)}
    }%
}

\Crefname{remark}{Remark}{Remarks}
\Crefname{remarkenumi}{Remark}{Remarks}
\AtBeginEnvironment{remark}{%
    \crefalias{enumi}{remarkenumi}%
    \setlist[enumerate,1]{
        label={\textit{\roman*)}},
        ref={\theremark.\roman*)}
    }%
}
\newcommand{\E}{\mathrm{E}}

\renewcommand{\>}{\rangle}

\newcommand{\R}{\mathbb{R}}
\newcommand{\Q}{\mathbb{Q}}
\renewcommand{\P}{\mathbb{P}}

\DeclareMathOperator{\Sg}{S}  
\DeclareMathOperator{\Tr}{Tr}  
\DeclarePairedDelimiterX\set[1]\lbrace\rbrace{\def\given{\;\delimsize\vert\;}#1} 

\begin{document}

\begin{abstract}
We consider the theory of bond discounts, defined as the difference between the terminal payoff of the contract and its current price. Working in the setting of finite-dimensional realizations in the HJM framework, under suitable notions of no-arbitrage, the admissible discount curves take the form of polynomial, exponential functions. We introduce reproducing kernels that are admissible under no-arbitrage as a tractable regression basis for the estimation problem in calibrating the model to market data. We provide a thorough numerical analysis using real-world treasury data.


\end{abstract}
\maketitle 

\section{Introduction}
In modelling the term structure of interest rates, it is standard practice to focus on either the instantaneous short rate or forward rates as the main building blocks for capturing the dynamics of interest rate evolution. Then, the no-arbitrage assumption results in the well-known relation between the zero-coupon bond prices and underlying rates:  bond prices can be expressed as exponential functions of the underlying interest rates. This paper focuses on the bond discount instead of working with instantaneous short or forward rates. The bond discount, defined as the difference between a bond’s terminal payoff and its current price, is our main object of study.

Our approach builds on the work presented in \cite{filipovic_discount}, providing a comprehensive bond discount theory.  Accordingly, starting from the Heath-Jarrow-Morton (HJM) framework (see \cite{hjm}), we develop a stochastic curve model for the bond discounts. By applying the Musiela parametrization (see \cite{musiela}), we can express the bond discount in terms of the solution to an infinite-dimensional stochastic differential equation. This equation is formulated within an appropriate Hilbert space of curves, providing a rigorous mathematical structure for the model (see, e.g., \cite{filipovic_phd}). A key condition for any viable financial model is the absence of arbitrage opportunities. In this context, we adopt the notion of no asymptotic free lunch with vanishing risk (NAFLVR) (see \cite{cuchiero_naflvr}), which requires that, in a large financial market, the semimartingales driving the market dynamics must be local martingales under the pricing measure. \cite{filipovic_discount} derives necessary and sufficient drift conditions for ensuring that NAFLVR holds within the discount framework. These conditions serve as a guiding principle for the development of our model.

To simplify the structure of the bond discounts, we assume a finite-dimensional affine geometry in the curve space. Under this assumption, the bond discount takes the form of the well-known quasi-exponentials (see, e.g., \cite{bjoerk}), leading to a tractable model well-suited for real-world applications. Similar to the methodology outlined in \cite{filipovic3}, we propose a statistical procedure for calibrating discount models using reproducing kernel Hilbert space (RKHS) techniques (for a standard material on RKHS, we refer to \cite{manton2015primer}). Specifically, the optimization procedure over the space of admissible curves is reduced to a finite-dimensional kernel regression problem, with a ridge regularization term included to ensure stability and robustness. We refer to the methodology proposed in \cite{filipovic_kr, filipovic3} for a slightly different alternative approach using kernel regression.

The problem of interpolating bond prices to derive an accurate term structure of interest rates is well-known in finance. Static interpolation schemes, which use a parametric family of curves, are frequently employed by banks and other financial institutions for inference. A prominent example of this approach is the Nelson-Siegel method (\cite{nelson1987parsimonious}), which parametrises the yield curve with a functional form designed to capture typical yield curve shapes. However, while static schemes like Nelson-Siegel are computationally efficient, they typically do not ensure dynamic consistency with the no-arbitrage condition (see, e.g., \cite{filipovic_phd}), making them unsuitable for more complex, arbitrage-free modelling.

Dynamic interpolation methods are more complex but necessary for arbitrage-free models in a dynamic setting. For instance, in \cite{jarrow_wu}, a cubic spline-based interpolation scheme is employed. The authors include additional functions to the set of cubic splines spanning the linear space containing the term structure to attain a dynamically consistent function space. More recent data-driven approaches, such as that proposed by \cite{autoencoder}, use autoencoders to interpolate term structures, aiming to stay close to a time-shift-invariant arbitrage-free manifold. See the references in \cite{autoencoder} for other data-driven methods.

In this paper, we make several significant contributions to term structure modelling in the bond discount framework.
We introduce a tractable model class based on a finite-dimensional affine specification, allowing for effective modelling of the discount curve while maintaining a practical structure suitable for real-world applications. Recognizing that the standard kernels used in \cite{filipovic3} do not ensure markets fulfilling no arbitrage for finite-dimensional affine models, we formulate an appropriate notion of so-called \textit{fully consistent kernels} that generate markets where contracts fulfill an HJM drift condition consistent with NAFLVR. After solving the mathematical reconstruction problem for fully consistent kernels, we characterize a rich family of RKHS containing discount curves that fulfill the no-arbitrage condition.

Building on this theoretical foundation, we validate our methodology empirically on real-world bond data, performing a day-by-day fitting procedure using our fully consistent kernels to capture the discount curve across the dataset.
As the next step, we attempt to calibrate a stochastic model consistent with the developed theory. This yields a tractable two-step procedure resulting in a fully parameterized model suitable for market predictions.  We present a detailed analysis of our numerical results to demonstrate the features of this approach.

The paper is structured as follows: In Section \Cref{sec:prelim} we provide the theoretical background for the discount framework. Section \Cref{sec:kernels} contains our definition of fully consistent kernels, which generate function spaces rich enough to contain non-trivial models fulfilling the NAFLVR condition. The section concludes with the formal statement of our main result in the form of fully consistent kernels for the discount framework. In Section \Cref{sec:calibration}, we calibrate our proposed model using US treasury (coupon) bond market data and perform a statistical analysis and interpretation of our results. We conclude with Section \Cref{sec:conclusion}. A theoretical background on Reproducing kernel Hilbert spaces (RKHS) is provided in \Cref{sec:rkhs_theory} and technical tools we employ in the paper in Appendix \Cref{app:b}. The more lengthy proofs of our main results are collected in Appendix \Cref{app:c}.
\section{Preliminaries}\label{sec:prelim}

\subsection{Notation}

Throughout this exposition, we will work on a stochastic basis $(\Omega ,\mathcal{F},(\mathcal{F}_t)_{t\geq 0},\Q)$ given by a filtered probability space with a right-continuous, complete filtration $(\mathcal{F}_t)_{t\geq 0}$ and measure $\Q$, which will play the role of the risk-neutral measure. Let $\P$ denote a measure such that $\P\ll\Q$, that is, $\P$ is absolutely continuous with respect to $\Q$, then we will denote by $\frac{d\P}{d\Q}$ the Radon-Nikodym derivative of $\P$ with respect to $\Q$.

Let $\R$ denote the set of real numbers. We will denote by $\R_+$ the subset of non-negative real numbers and for $d\in\mathbb N$, $\R^d$ the $d$-dimensional Euclidean space. Given a vector $u\in\R^d$, we will write $u^{\top}=(u_1,\dots ,u_d)$ to denote its components, where $u^{\top}$ is the vector transpose of $u$. For two elements $v,w\in\R^d$, the Euclidean scalar product will be written as $\langle v,w\rangle = v^{\top}w$. We will make use of the so-called extended vector notation: given $d\in\mathbb N$, we will call $\R^{d+1}$ the extended vector space and start indexing the first coordinate with $0$, that is, for $u\in\R^{d+1}$, we will write $u=(u_0,\dots ,u_d)^{\top}$. Furthermore, given a vector $v\in\R^d$ and $v_0\in\R$, we will write $(v_0,v)^{\top}\in\R^{d+1}$. We will use the same notational conventions for matrices. The identity matrix of dimension $d$ will be denoted by $\mathbbm 1_d\in\R^{d\times d}$. Let $A\subseteq\R^d$ be a set, then $\text{aff}(A)\subseteq\R^d$ denotes the affine hull of $A$.

Let $f:\R^n\rightarrow\R^m$, $y\mapsto f(y)$ be a smooth function. We will denote by $\partial _{y_k}f$ the partial derivative of $f$ with respect to the coordinate $y_k$ for $k=1,...,n$. We will use the notation $D_yf=(\partial _{y_i}f_j)_{i=1\dots,n, j=1,\dots,m}$ for the Jacobian matrix of $f$ with respect to the variable $y$ and for $k=1,...,m$, $D^2_yf_k=(\partial _{y_i}\partial _{y_j}f_k)_{i,j=1,\dots, n}$ will denote the Hessian matrix and $\nabla _yf_k=(\partial _{y_1}f_k,\dots ,\partial _{y_n}f_k)^{\top}$ will denote the gradient of the $k$-th component of $f$ with respect to the variable $y$.

Let $(\mathcal{H},\lVert\cdot\rVert _{\mathcal H})$ be a separable Hilbert space of functions from $\mathbb R_+$ to $\mathbb R$ fulfilling
\begin{itemize}
	\item[\textbf{(H1)}] $\delta_0:\mathcal{H}\rightarrow\mathbb R, h\mapsto h(0)$ is continuous linear.\customlabel{H1}{(H1)}
	\item[\textbf{(H2)}] $\mathcal \Sg_h:\mathcal{H}\rightarrow \mathcal{H}, f\mapsto (x\mapsto f(x+h))$ defines a $c_0$-semigroup $\mathcal (\Sg_h)_{h\geq 0}$ on $\mathcal{H}$ whose generator will be denoted by $\partial_x$, cf.\ \cite[pp.\ 6--8]{EK}.\customlabel{H2}{(H2)}
\end{itemize}
The inner product on $\mathcal{H}$ will be denoted $\<h_1,h_2\>_{\mathcal H}$ for $h_1,h_2\in \mathcal{H}$. We will make use of a Hilbert space $\mathcal H$ satisfying Assumptions \ref{H1} and \ref{H2} throughout the paper without explicitly referencing them. An example of a Hilbert space with those properties are the forward curve spaces $\mathcal H_w$ introduced in \cite{filipovic_consistency}.
\subsection{Model description}
In the following,  we will work with the zero-coupon \emph{bond discount} curve where the bond discount for maturity $T$ at time $t$ is defined as the difference between the corresponding bond's face value, 1 and its present value at time $t$. Let $P(t,T)$ denote the price at time $t$ of a zero-coupon bond with a maturity date $T$. We denote the corresponding bond discount by
\begin{equation}\label{eq:4_1}
	H_t(T-t)=1-P(t,T).
\end{equation}
For our purposes, we will model $H$ as a diffusion process $H:\R_+\times\Omega\rightarrow\mathcal H$, $H:(t,\omega)\mapsto H_t(\omega)$, that is, $H$ takes values in the Hilbert space $\mathcal H$. Let $W$ be a d-dimensional Brownian motion in  $(\Omega ,\mathcal{F},(\mathcal{F}_t)_{t\geq 0},\Q)$. We assume $H_t$ satisfies
\begin{equation}\label{eq:4_2}
	H_t=\Sg_tH_0+\int _0^t\Sg_{t-s}\alpha _sds+\int _0^t\Sg_{t-s}\Sigma _sdW_s,
\end{equation}
for an appropriate drift coefficient $\alpha :\R_+\times\Omega\rightarrow\mathcal H$ and diffusion coefficient $\Sigma :\R_+\times\Omega\rightarrow L(\R^d,\mathcal H)$. Here, once again $\Sg_h:\mathcal H\rightarrow \mathcal H$, $f\mapsto f(\cdot +h)$ denotes the shift operator of the semigroup of left shifts whose generator will be denoted by $\partial _x$. This implies $H_t$ is the mild solution to the stochastic differential equation (see, e.g., \cite{zabczyk})
\begin{equation}\label{eq:4_sde}
	dH_t=\left(\partial _xH_t+\alpha _t\right)dt+\Sigma _tdW_t.
\end{equation}
\begin{remark}
    Consider now the forward curve spaces $\mathcal H_w$ defined in \cite[Definition 5.1.1]{filipovic_consistency}. Since $\mathcal H_w$ is a normed space, we have by \ref{H1} that the evaluation functional $\delta _x$ is bounded, hence $\mathcal H_w$ is a RKHS. This fact is leveraged in e.g. \cite{filipovic3, filipovic_kr} where the authors derive a reproducing kernel for $\mathcal H_w$ and use its properties for an efficient interpolation scheme of the discount curve.
\end{remark}
\subsection{HJM-drift condition}
 This section includes the results on the sufficient conditions for NAFLVR. The sufficiency conditions were first proven in \cite{filipovic_discount}, where they are formulated in terms of the evolution of a random field. We adopt the setting of the Musiela parametrization (\cite{musiela}) and state the corresponding results here for convenience and completeness. 

Assume the limit $r_t:=\lim _{T\rightarrow t}-\partial _T\log (P(t,T))$ exists. In that case, $r_t$ is the short rate, and it follows immediately from \Cref{eq:4_1} that $r_t=H_t'(0)$. The discounted price of the zero-coupon bond is given by 
\begin{equation}\label{eq:4_3}
	\tilde{P}(t,T)=e^{-\int _0^tr_sds}P(t,T).
\end{equation} 
We now proceed with deriving the HJM-drift condition.
\begin{proposition}\label{prop:4_1}
	Assume that the process $H$ as defined in \eqref{eq:4_2} is additionally a strong solution to the SDE \eqref{eq:4_sde}. Let $\beta _t(x):=\delta _x(\partial _xH_t+\alpha _t)$ denote its drift coefficient then the processes $(\tilde{P}(t,T))_{t\in [0,T]}$ as defined in \Cref{eq:4_3} are local martingales for all $T\in\R_+$ if and only if the drift $\beta _t(x)$ of $H_t(x)$ fulfills for all $x,t>0$, $\Q$-a.s.
	\begin{equation}\label{eq:4_5}
		\beta _t(x)=H'_t(x)-r_t+r_tH_t(x).
	\end{equation}
\end{proposition}
\begin{proof}
	See \Cref{app:b}.
\end{proof}
\subsection{Pricing under the forward measure}

One way to facilitate the methodology developed in the discount setting is to use it to price interest rate derivatives on the market. To do this, one may use \emph{forward pricing measures}, which are important tools for simplifying the pricing of interest rate derivatives, such as caps and floors, swaptions, and bond options, to name a few. Forward measures are closely tied to the concept of numeraires, which are benchmark assets used to express prices in relative terms.

The $T$ forward pricing measure is a probability measure under which the price of a zero-coupon bond maturing at time $T$ is deterministic. Moreover, all discounted asset prices relative to the $T$ bond are (local) martingales under that measure. The measure transformation from the original risk-neutral measure (e.g., the  $\Q$ measure) to the forward measure is achieved via the Radon-Nikodym derivative. In the following, we make this notion precise and provide the form of the density process for the change of measure in terms of the dynamics of the discount process.

Let $B_t:=e^{\int _0^tr_sds}$ denote the risk-free bank account with initial condition $B_0=1$. The forward measure $\mathbb P^T$ is defined as the measure equivalent to the risk-neutral measure $\Q$ with the Radon-Nikodym derivative (cf. \cite[Definition 9.6.2.]{musiela_rutkowski})
\begin{equation*}
	\frac{d\mathbb P^T}{d\Q}=\frac{1}{B_TP(0,T)}.
\end{equation*}
\begin{proposition}
	Let $\Q$ denote the risk-neutral measure and $\mathbb P ^T$ be the forward measure. Then the process
	\begin{equation*}
		\eta _t:=\left.\frac{d\mathbb P ^T}{d\Q}\right\vert _{\mathcal F_t}
	\end{equation*}
	fulfills
	\begin{equation*}
		\eta _t=\mathcal E_t\left(\int _0^{\cdot}\frac{\Sigma _s(T-s)}{1-H_t(T-s)}dW_s\right).
	\end{equation*}
    where $\mathcal E_t(X)$ denotes the stochastic exponential of $X$ (cf. \cite[Section I.4f.]{jacod}).
\end{proposition}
\begin{proof}
	By Girsanov's Theorem, we have $\eta _t=\mathcal E_t(\lambda)$ for some adapted process $\lambda$. We have, by definition of the stochastic exponential,
	\begin{equation*}
		d\eta _t=\eta _td\lambda _t
	\end{equation*}
	and therefore
	\begin{equation*}
		d\lambda _t=\frac{1}{\eta _t}d\eta _t.
	\end{equation*}
	On the other hand, by the definition of the forward measure, we have
	\begin{equation*}
		\eta _t=\E ^{\Q}\left[\left.\frac{1}{B_TP(0,T)}\right\vert\mathcal F_t\right] =\E ^{\Q}\left[\left.\frac{P(T,T)}{B_TP(0,T)}\right\vert\mathcal F_t\right]=\frac{P(t,T)}{B_tP(0,T)}.
	\end{equation*}
	Thus, we obtain
	\begin{equation*}
		\begin{aligned}
                d\lambda _t=&\frac{B_tP(0,T)}{P(t,T)}d\left(\frac{P(t,T)}{B_tP(0,T)}\right) =\frac{1}{P(t,T)}\left( dP(t,T)-r_tP(t,T)dt\right) \\=&\frac{1}{1-H(T-t)}\left( \alpha _t(T-t)dt+\Sigma _t(T_t)dW_t -r_t(1-H_t(T-t))dt\right).
		\end{aligned}
	\end{equation*}
	Since the discounted process $B_t^{-1}P(t,T)$ is a local martingale under $\Q$, we may use the drift condition of \Cref{prop:4_1} to obtain
	\begin{equation*}
		d\lambda _t=\frac{\Sigma _t(T-t)}{1-H_t(T-t)}dW_t
	\end{equation*}
	and the assertion follows from the definition of the stochastic exponential.
\end{proof}

\subsection{Linearity assumption}
This section assumes a more tractable structure for the discount process. In particular, we will require the solutions to lie in a finite-dimensional affine subspace. This will force the curve itself to belong to the exponential-affine family and ensure that the (finite-dimensional) stochastic driving process fulfills a slightly modified quadratic drift condition. We make the linearity assumption precise.
\begin{itemize}
    \item[\textbf{(LA)}] \customlabel{LA}{(LA)} 
	We assume that the discount process $H$ satisfies  $H_t(x)=g_0(x)+\sum _{i=1}^dg_i(x)f_i(Y_t)$, where 
	\begin{enumerate}
		\item $g_0,\dots,g_d\in C^1(\R_+,\R)$ and $f_1,\dots,f_d\in C^2(\R ^k,\R )$, satisfy 
        \begin{equation*}
            \text{aff}(\lbrace\! (g_1(x),\dots,g_d(x))^{\top},\, x \geq 0\rbrace)=\R ^d 
        \end{equation*}
        and
        \begin{equation*}
            \text{aff}(\lbrace f_1(y),\dots,f_d(y), y\in\R ^d\rbrace )=\R ^d.
        \end{equation*}
		\item $Y$ is a $d$-dimensional diffusion process with dynamics 
        \begin{equation*}
            Y_t=Y_0+\int _0^tb_sds+\int _0^t\sigma _sdW_s,
        \end{equation*}
        where $Y_0$ is an $\mathcal{F}_0$-measurable random variable in $\R ^d$, $b:\R_+\times\Omega\rightarrow\R^k$ is a progressively measurable $d$-dimensional process with almost surely integrable paths and $\sigma:\R_+\times\Omega\rightarrow\R^{d\times k}$ is a progressively measurable $d$-dimensional matrix process with almost surely integrable paths. 
	\end{enumerate}
	Henceforth (with a slight abuse of notation), we shall set $g:=(g_0,\dots ,g_d)^{\top}$, as well as $f:=(1,f_1,\dots ,f_d)^{\top}$ and write 
    \begin{equation*}
        H_t(x)=\langle g(x),f(Y_t)\rangle .
    \end{equation*}
\end{itemize}

\begin{example}
	We collect here some of the example models that satisfy \ref{LA}.
	\begin{enumerate}
		\item Linear-rational models (see e.g. \cite{filipovic_lrtsm}): 
        \begin{equation*}
        H_t(x)=\frac{\langle g(x), (1,Y_t)\rangle}{\langle \lambda ,(1,Y_t)\rangle}
        \end{equation*}
        for $\lambda\in\mathbb{R}^{d+1}$ satisfies $H_t(x)=\langle g(x), f(Y_t)\rangle$ with 
        \begin{equation*}
            f(y):=\frac{(1,y)^{\top}}{\langle \lambda ,(1,y)\rangle}.
        \end{equation*}
		\item Polynomial models: 
        \begin{equation*}
            H_t(x)= \phi_0(x)+\sum _{\vert\alpha\vert =1}^{n}\phi _{\alpha}(x)Y_t^{\alpha}
        \end{equation*} satisfies $H_t(x)=\langle g(x),f(Y_t)\rangle$ with 
        \begin{equation*}
            g(x):=(\phi _0(x), \phi _1(x),\dots ,\phi _d(x),\phi _{11}(x),\dots ,\phi _{1d}(x),\dots ,\phi _{dd}(x),\dots )^{\top}
        \end{equation*} 
        and 
        \begin{equation*}
            f(y):=(1,y_1,\dots ,y_d,y_1y_1,\dots ,y_1y_d,\dots ,y_dy_d,\dots )^{\top}.
        \end{equation*}
	\end{enumerate}
\end{example}

In the following, we assume that the process $H$ satisfies \ref{LA}. We will see that under the drift condition, this imposes structure in the form of $g$ and $f(Y)$. 
\begin{proposition}\label{prop:4_3}
	Let $H$ be a discount process satisfying \ref{LA} and $P(t,T)=1-H_t(T-t)$ be the zero-coupon bond model induced by $H$. Then the discounted bond price processes $(\tilde P(t,T))_{0\geq t\geq T}$ are local martingales for all $T>0$ if and only if
	\begin{enumerate}
		\item There exists a matrix $M\in\R^{(d+1)\times (d+1)}$ such that the function $g$ satisfies
			\begin{equation}\label{eq:4_prop23b}
				g(x)=\left(\mathbbm{1}_{d+1}-e^{xM}\right) e_0,
			\end{equation}
			where $e_0=(1,0,\dots ,0)^{\top}\in\R ^{d+1}$ is the first unit basis vector
		\item The drift and diffusion coefficients $b$, respectively $\sigma$ of the process $Y$ satisfy $\Q$-a.s.
			\begin{equation}\label{eq:4_prop23c}
                    \begin{aligned}
				    &D_yf(Y_t)(0,b_t)^{\top}+\frac12\sum _{k=1}^d\Tr\left( \sigma _t\sigma _t^{\top}D^2_yf_k(Y_t)\right) e_k\\
                    &=(M^{\top}+\langle g'(0),f(Y_t)\rangle\mathbbm{1}_{d+1})f(Y_t).
                    \end{aligned}
			\end{equation}
	\end{enumerate}
\end{proposition}

By assuming $f=(1,id)$, we may also solve for the drift term of the process $Y$ explicitly. Indeed, by redefining $X_t=f(Y_t)$, we can always reduce to the following simplified case.
\begin{corollary}
	In the case $f_i=id$ for $i=1,\dots ,d$, \Cref{eq:4_prop23c} simplifies greatly and we obtain
	\begin{equation}\label{eq:4_remark1}
		(0,b_t)^{\top}=\left( M^{\top}+\langle g'(0),(1,Y_t)\rangle\mathbbm{1}_{d+1}\right) (1,Y_t)^{\top},
	\end{equation}
	that is, we may solve for the drift directly.
\end{corollary}

\begin{remark}
	From the drift condition, it is clear that a discount model fulfilling NAFLVR is induced by specifying the matrix that induces the curve component and the diffusion matrix that induces the stochastic component. Indeed, an arbitrage-free model is fully specified, up to reparametrisation, by the triplet $(M, y_0, \sigma )$, where $Y_0=y_0$ is the starting value of the process $Y$. 
\end{remark}
The following result implies that it is sufficient to consider models in a convenient basis transformation.
\begin{corollary}\label{cor:4_basis_change}
	Let $M\in\R ^{(d+1)\times (d+1)}$ and $\sigma :\R _+\times\Omega\rightarrow\R ^{d\times d}$ be a progressively measurable $d$-dimensional matrix process with a.s. integrable paths and let $H$ be the finite-dimensional affine discount model induced by the triplet $(M, y_0, \sigma)$ which fulfills the no-arbitrage condition in the sense of \Cref{prop:4_1}, then $H$ satisfies
	\begin{equation*}
		H_t(x) = 1-\langle e^{xJ}p,Z_t\rangle,
	\end{equation*}
	where $Z:=P^{\top}(1,Y)^{\top}$, $p=P^{-1}e_0$ and $P,J\in\R^{(d+1)\times (d+1)}$ are such that $M=PJP^{-1}$ is the Jordan decomposition of $M$. Furthermore, $Z$ satisfies the drift condition
	\begin{equation}\label{eq:4_drift_basis_change}
		b^Z_t=\left( J+\langle Jp,Z_t\rangle\mathbbm 1_{d+1}\right) Z_t,
	\end{equation}
	where $b^Z$ is the drift of $Z$.
\end{corollary}
\begin{proof}
	We have by definition
	\begin{equation*}
		H_t(x)=\langle g(x),(1,Y_t)\rangle.
	\end{equation*}
	Since $H$ fulfills the no-arbitrage condition, we obtain
	\begin{equation*}
		\begin{aligned}
			H_t(x)&=\langle g(x), (1,Y_t)\rangle = \langle (\mathbbm 1_{d+1}-e^{xM})e_0, (1,Y_t)\rangle =1-\langle e^{xM}e_0,(1,Y_t)\rangle.
		\end{aligned}
	\end{equation*}
	By the properties of the matrix exponential and the Jordan decomposition, we have
	\begin{equation*}
		H_t(x)=1-\langle Pe^{xJ}P^{-1}e_0, (1,Y_t)\rangle = 1-\langle e^{xJ}p,P^{\top}(1,Y_t)\rangle =1-\langle e^{xJ}p, Z_t\rangle.
	\end{equation*}
	Let $b_t$ denote the drift of the process $Y$. The drift condition \eqref{eq:4_remark1} implies
	\begin{equation*}
		\begin{aligned}
			(0,b_t)^{\top}&=\left( (PJP^{-1})^{\top}+\langle g'(0),(1,Y_t)\rangle\mathbbm 1_{d+1}\right) (1,Y_t)^{\top}\\
			   &= (P^{-1})^{\top}JP^{\top}(1,Y_t)^{\top}+\langle g'(0),(1,Y_t)\rangle (1,Y_t)^{\top}\\
			   &= (P^{-1})^{\top}JZ_t+\langle g'(0), (P^{-1})^{\top}Z_t\rangle (P^{-1})^{\top}Z_t\\
			   &=(P^{-1})^{\top}\left( J+\langle P^{-1}g'(0), Z_t\rangle\right) Z_t\\
			   &=(P^{-1})^{\top}\left( J+\langle P^{-1}Me_0, Z_t\rangle\right) Z_t\\
			   &=(P^{-1})^{\top}\left( J+\langle Jp, Z_t\rangle\right) Z_t.
		\end{aligned}
	\end{equation*}
	Multiplying by $P^{\top}$ on both sides and noting that $P^{\top}(0,b_t)^{\top}$ is the drift of the process $Z_t=P^{\top}(1,Y_t)^{\top}$ completes the proof.

\end{proof}
\begin{corollary}\label{cor:4_diagonalisable_m}
	Given the same assumptions as in \Cref{cor:4_basis_change}, assume additionally that $M$ is diagonalizable over $\R$. Then $H$ satisfies
	\begin{equation}\label{eq:4_diag_m}
		H_t(x)=1-\langle e^{xD}\bm{1}, Z_t\rangle,
	\end{equation}
	where $\bm 1\in\R ^{d+1}$ is the vector of ones, $D=\text{diag}(\lambda )$ is the diagonal matrix generated by the vector $\lambda$ of the eigenvalues $\lbrace\lambda _0,\dots ,\lambda _d\rbrace$ of $M$ and $Z$ is defined as $\text{diag}(p)P^{\top}(1,Y_t)^{\top}$, where $p=P^{-1}e_0$. Furthermore, $Z$ satisfies the drift condition
	\begin{equation}\label{eq:4_drift_cond_diag_m}
		b_t^Z=(D+\langle \lambda,Z_t\rangle \mathbbm 1_{d+1})Z_t,
	\end{equation}
	where $b^Z$ is the drift of $Z$.
\end{corollary}
\begin{proof}
	This follows directly by observing that in the case of a diagonal matrix $D$, we may use coordinate-wise multiplication of $p$.
\end{proof}
\begin{remark}
	It is easy to check that $g'(x)=e^{xM}m$, where $m\in\R^{d+1}$ is the first column vector of $M$, that is,  $m_i=M_{i,0}$ for $i=0,\dots ,d$.
\end{remark}

Thus, given the simplifying Assumption \ref{LA}, we may directly determine the space of admissible curves for a risk-neutral HJM model of discount curves, namely the space of quasi-exponentials. We can also determine conditions on the process $Y$ in this framework. 
\subsection{Time inhomogeneous case}
Certain non-trivial extensions of the affine discount model can be considered for more advanced applications. For completeness, we will provide one simple way of generalizing to the time-inhomogeneous case and prove admissibility conditions. This extended model's theoretical implications and statistical practicality are left for future research.

In the following, we make a slight generalization by letting the function $g$ depend on the time of observation $t$; that is, we set $H_t(x)=\langle g(x,t),f(Y_t)\rangle$. Through the application of It\^o's Lemma, we get that the resulting drift of $H(x)$ is of the form
\begin{equation}\label{eq:4_drift_time_inhom}
	\beta _t(x)=\langle\partial _tg(x,t),f(Y_t)\rangle+\langle g(x,t),D_yf(Y_t)b_t+\frac12\sum _{k=1}^d\Tr\left(\sigma _t\sigma _t^{\top}H_yf_k(Y_t)\right) e_k\rangle.
\end{equation}
To get an analogous result to \Cref{prop:4_3}, we first observe that the drift condition \eqref{eq:4_5} implies the same condition on the drift and diffusion coefficients of $Y_t$ as in \Cref{eq:4_prop23c}. On the other hand, the function $g$ is now obtained by solving a partial differential equation. We summarize the corresponding result in the following.
\begin{proposition}\label{prop:4_time_inhom_case}
	Let $H$ be a discount process satisfying \ref{LA} for a function $g\in C^{1,2}(\R _+\times\R _+,\R ^{d+1})$ and let $P(t,T)=1-H_t(T-t)$ be the zero-coupon bond model induced by $H$. Then the discounted bond price processes $(\tilde{P}(t,T))_{t\in [0,T]}$ are local martingales for all $T>0$ if and only if
	\begin{enumerate}
		\item There exists a matrix function $M:\R_+\rightarrow\R^{(d+1)\times (d+1)}$ such that the function $g$ is the solution to the problem
	\begin{equation}\label{eq:4_prop_time_inhom1}
		\begin{aligned}
			&\partial _xg(x,t)-\partial _tg(x,t) = M(t)g(x,t)+\partial _xg(0,t),\\
			&g(0,t)=0,
		\end{aligned}
	\end{equation}
	\item The drift and diffusion coefficients $b$, respectively $\sigma$ of $Y$ satisfy
	\begin{equation}\label{eq:4_prop_time_inhom2}
            \begin{aligned}
		      &D_yf(Y_t)(0,b_t)^{\top}+\frac12\sum _{k=1}^d\Tr\left( \sigma _t\sigma _t^{\top}H_yf_k(Y_t)\right) e_k \\
                &=(M(t)^{\top}+\langle \partial _xg(0,t),f(Y_t)\rangle\mathbbm{1}_{d+1})f(Y_t).
            \end{aligned}
	\end{equation}
		\end{enumerate}
\end{proposition}
\begin{proof}
	See \Cref{app:b}.
\end{proof}
Assuming some regularity on the matrix function $M(t)$, we may obtain an explicit solution.
\begin{corollary}\label{cor:4_time_inhom_expl}
	Assume that $M(t_1)M(t_2)=M(t_2)M(t_1)$ for all pairs $(t_1,t_2)\in [0,T]\times[0,T]$. Then, the function $g$ is of the form
	\begin{equation}\label{eq:4_pde_expl_sol}
		g(x,t)=e^{\int _0^xM(T-\xi)d\xi}\int _0^xe^{-\int _0^{\xi}M(T-\zeta )d\zeta}\partial _xg(0,T-\xi )d\xi.
	\end{equation}
\end{corollary}
\begin{proof}
	See \Cref{app:b}.
\end{proof}
\section{Fully consistent kernels and induced RKHS}\label{sec:kernels}
Having derived the HJM-type conditions for the discount such that the induced bond market fulfills the NAFLVR condition, we are interested in finding kernels that generate RKHS rich enough to contain such markets. Indeed, we will verify that suitable kernels can generate models fulfilling the NAFLVR condition. We will make precise the notion of these fully consistent kernels. Before we give an appropriate notion, we begin with a definition for kernels, which hold a special significance for the rest of our considerations.
\begin{definition}\label{def:pol_exp_kernels}
	\begin{enumerate}
		\item Let $p:\R _+\times\R _+\rightarrow\R$ be a symmetric positive semidefinite polynomial in the sense of \Cref{def:4_kernel_function} (for a discussion of the terminology of positive semidefinite functions, see e.g. \cite[Section 2.2]{paulsen_rkhs}). Then $k_p(x,y):=p(x,y)$ is a kernel function, which we will refer to as the \emph{polynomial kernel} with the induced space $\mathcal H(p)$.
		\item Let $k_{\exp}(x,y):=e^{xy}$. We will refer to $k_{\exp}$ as the \emph{exponential kernel} with induced RKHS $\mathcal H(\exp)$.
	\end{enumerate}
\end{definition}
Note that we distinguish between the kernels and the functions themselves to avoid confusion whenever functions are used in a different context. In the definition of the respective RKHS, we, however, omit this distinction to emphasize the nature of the kernel. By a slight abuse of notation, we will use the same notation regardless of the reparametrization of the arguments of the functions and scaling by constant factors; that is, we will still call, e.g., $k(x,y)=ce^{(ax+b)(ay+b)}$ an exponential kernel.
	
In the next step, we introduce the kernels that generate spaces that are consistent with the admissibility conditions for the discount model. To this extent, we will use the following definition.
\begin{definition}\label{def:fully consistent_kernel}
	Let $k$ be a kernel function on $\R _+$ in the sense of \Cref{def:4_kernel_function}. We say $k$ is \emph{fully consistent} if for any choice $y_1,\dots,y_N\in\R _+$, $N\in\mathbb N$ there is a finite-dimensional space $V\subseteq C^1(\R _+, \R)$ fulfilling
	\begin{enumerate}
		\item $\partial _x(V)\subseteq V$.\label{def:fully consistent_kernel_1}
		\item $k_{y_1},\dots,k_{y_N}\in V$.\label{def:fully consistent_kernel_2}
	\end{enumerate}
\end{definition}
\begin{remark}
	Condition \ref{def:fully consistent_kernel_1} of \Cref{def:fully consistent_kernel} is motivated by the theory of invariant manifolds (see, e.g., \cite{filipovic_invariant, teichmann1, tappe_invariant}). Indeed, the choice of $V$ as the smallest finite-dimensional derivative invariant space, which contains the span of kernels, is the natural choice of domain for the discount process. In some practical cases, this space coincides with the span of the kernels, e.g., in the case of the exponential kernel. This justifies the use of kernels fulfilling the conditions of \Cref{def:fully consistent_kernel} as a regression basis for the statistical calibration of the model.
\end{remark}
Next, we state a key result that will yield the necessary and sufficient conditions for the full consistency of kernels. This will serve as the basis for deriving kernels, which we will use as a regression basis for the estimation problem of the discount models.
\begin{proposition}\label{prop:4_u_space}
	Let $\mathcal{U}:=\lbrace x\mapsto \varphi\left(\left(\mathbbm{1}_{d+1}-e^{xM}\right) e_0\right): M\in\R ^{(d+1)\times (d+1)}, \varphi\in (\R ^{d+1})^*, d\in\mathbb{N}\rbrace$. A kernel $k$ for some RKHS is fully consistent if and only if the map $k_y:x\mapsto k(y,x)$ is such that $k_y\in\mathcal{U}$ for all $y\in\R _+$.
\end{proposition}
\begin{proof}
	See \Cref{app:b}.
\end{proof}
The space $\mathcal{U}$ can be seen as the space of admissible functions where we only consider a specific coordinate since we can choose $\varphi =\langle\cdot ,e_k\rangle$ for $k=1,\dots ,d$. We additionally state the following Lemma, which asserts a ``nice'' vector space structure on the set $\mathcal U$. In particular, this allows us to combine fully consistent kernels to suit our practical needs.
\begin{lemma}\label{lem:v_space_u}
	$\mathcal{U}$ is a vector space of $C^{\infty}$-functions.
\end{lemma}
\begin{proof}
	See \Cref{app:b}.
\end{proof}
\subsection{RKHS induced by fully consistent kernels}
\Cref{prop:4_u_space} gives a general condition for the full consistency of the kernel functions. We will now consider specific examples of kernels that can be used for the discount model and derive descriptions for the RKHS induced by those kernels. Indeed, we will show that there is a broad class of fully consistent kernels with induced RKHS, which provide a rich modelling basis for the discount model while retaining tractability for the computational task of calibrating to the market data. We will start with a definition.
\begin{definition}\label{def:weighted_norm}
	Let $\ell ^2(\R)$ denote the Hilbert space of square-summable sequences over $\R$ and let $w=(w_k)_{k\geq 0}\subseteq\R _+$ denote a sequence (possibly not in $\ell ^2(\R)$). Define for $f,g\in\ell ^2(\R)$ the weighted product
	\begin{equation}\label{eq:4_weighted_prod}
		\langle f,g\rangle _{\ell _w^2}=\langle w\odot f, g\rangle _{\ell ^2}=\langle f,w\odot g\rangle _{\ell ^2},
	\end{equation}
	where $w\odot f=(w_kf_k)_{k\geq 0}$ denotes the Hadamard product, and the weighted norm
	\begin{equation}\label{eq:4_weighted_norm}
		\Vert f\Vert ^2_{\ell ^2_w}:=\langle f,f\rangle _{\ell ^2_w}.
	\end{equation}
\end{definition}
We can now state a general result that will enable us to characterize fully consistent kernels for our model and obtain the RKHS induced by those kernels.
\begin{lemma}\label{lem:rkhs_general}
	Let $h:\R\rightarrow\R$ be a real-analytic function with $h^{(k)}(0)\geq 0$, where $h^{(k)}$ denotes the $k$-th order derivative of $h$ for $k\in\mathbb{N}$. Define the weight sequence $w=(w_k)_{k\geq 0}$ where $w_k:=\mathbbm{1}_{\lbrace h^{(k)}(0)>0\rbrace}(h^{(k)}(0))^{-1}$. Let $a,c\in\R$, $a\neq 0$ and define $k(x,y):=h((ax-c)(ay-c))$. Then $k$ is the reproducing kernel of the RKHS
	\begin{equation}\label{eq:4_rkhs_general_description}
        \mathcal{H}(k)=\set[\Bigg]{f:\R\rightarrow\R\given f(x)=\sum _{k=0}^{\infty}\mathbbm{1}_{\lbrace h^{(k)}(0)>0\rbrace}b_k\left(x-\frac{c}{a}\right)^k, \left\Vert (b_k)_{k\geq 0}\right\Vert _{\ell ^2_w}\!\!<\infty},
	\end{equation}
	where $b_k:=f^{(k)}(c/a)/k!$, with inner product $\langle\cdot ,\cdot\rangle _{\mathcal{H}(k)}$ given by
	\begin{equation}\label{eq:4_rkhs_general_inner_prod}
		\langle f,g\rangle _{\mathcal{H}(k)}=\sum _{k=0}^{\infty}\frac{w_k}{a^{2k}k!}f^{(k)}\left(\frac{c}{a}\right) g^{(k)}\left(\frac{c}{a}\right)
	\end{equation}
	and with the induced norm
	\begin{equation}\label{eq:4_rkhs_general_inner_norm}
		\Vert f\Vert ^2_{\mathcal{H}(k)}=\sum _{k=0}^{\infty}\frac{w_k}{a^{2k}k!}\left\vert f^{(k)}\left(\frac{c}{a}\right)\right\vert^2.
	\end{equation}
\end{lemma}
\begin{proof}
	See \Cref{app:b}.
\end{proof}
Given the results of \Cref{prop:4_u_space}, we may choose the rich class of fully consistent kernels from the exponential-affine family of functions, and using the results of \Cref{lem:rkhs_general}, we obtain a precise description of the RKHS induced by those kernels.
\begin{proposition}\label{prop:4_rkhs}
	Let $p(t):=a_dt^d+...+a_1t+a_0$ be a polynomial, such that $a_k\geq 0$ for $k=0,\dots ,d$. Let $h_k:=e^{-\alpha ^2/\beta}\sum _{l=0}^{k\wedge\emph{deg}(p)}\binom{k}{l}l!a_l$ for $k\in\mathbb N_0$ and define the weight sequence $w=(w_k)_{k\geq 0}$, where $w_k:=\mathbbm{1}_{\lbrace h_k>0\rbrace}h_k^{-1}$. Let $\alpha,\beta\in\R$, $\alpha\geq 0,\beta > 0$ and define $k(x,y):=p((\sqrt{\beta}x-\alpha /\sqrt{\beta} )(\sqrt{\beta}y-\alpha /\sqrt{\beta}))e^{\beta xy-\alpha (x+y)}$. Then $k$ is a fully consistent kernel in the sense of \Cref{def:fully consistent_kernel} and is the reproducing kernel of the RKHS
	\begin{equation}\label{eq:4_rkhs_description}
        \mathcal{H}(k)=\set[\Bigg]{ f:\R _+\rightarrow\R\given f(x)=\sum _{k=0}^{\infty}\mathbbm{1}_{\lbrace h_k>0\rbrace}
        b_k\left(x-\frac{\alpha}{\beta}\right)^k,\left\Vert (b_k)_{k\geq 0}\right\Vert _{\ell ^2_w} <\infty },
	\end{equation}
	where $b_k:=f^{(k)}(\alpha /\beta)/k!$, with inner product $\langle\cdot ,\cdot\rangle _{\mathcal{H}(k)}$ given by
	\begin{equation}\label{eq:4_inner_prod}
		\langle f,g\rangle _{\mathcal{H}(k)}=\sum _{k=0}^{\infty}\frac{w_k}{\beta ^kk!}f^{(k)}\left(\frac{\alpha}{\beta}\right) g^{(k)}\left(\frac{\alpha}{\beta}\right),
	\end{equation}
	and induced norm
	\begin{equation}\label{eq:4_rkhs_norm}
		\Vert f\Vert ^2_{\mathcal{H}(k)}=\sum _{k=0}^{\infty}\frac{w_k}{\beta ^kk!}\left\vert f^{(k)}\left(\frac{\alpha}{\beta}\right)\right\vert ^2.
	\end{equation}
\end{proposition}
\begin{proof}
	See \Cref{app:b}.
\end{proof}
\begin{remark}
	\begin{enumerate}
		\item In the case $\beta =0$, $p\equiv 1$, the function $k(x,y):=e^{-\alpha (x+y)}$ is a fully consistent kernel and induces the reproducing kernel Hilbert space $\mathcal{H}(k)=\lbrace f:\R_+\rightarrow\R\vert f(x)=ce^{-\alpha x}, c\in\R\rbrace$ with inner product $\langle c_1e^{-\alpha\cdot},c_2e^{-\alpha\cdot}\rangle _{\mathcal{H}(k)} = c_1c_2$. This follows from the fact that $k(x,y)=f(x)f(y)$, where $f(t):=e^{-\alpha t}$ and the considerations in \cite[Proposition 2.19]{paulsen_rkhs}.
		\item One may extend the definition of the kernel $k$ to the case $\alpha\in\R$. The shift in the Taylor-series to the negative results in a reproducing kernel Hilbert space $\mathcal{H}(k)$ of real analytic functions $f:\R\rightarrow\R$ which fulfill the same weighted square-summability condition. Since we are only interested in term structure models in positive time, we may consider the restriction $f\big\vert _{\R_+}$ for $f\in\mathcal{H}(k)$.
		\item In the case $p\equiv 1$, the resulting RKHS is a special case of the Segal-Bargmann space of analytic functions (see, for instance, \cite[Subsection 7.3.2]{paulsen_rkhs} and \Cref{prop:4_bargmann_space}).
		\item Let $\varphi :x\mapsto \sqrt{\beta}x-\alpha /\sqrt{\beta}$ and $\Delta :x\mapsto (x,x)$. Define $\tilde{k}(x,y):=xy, k_p:=p\circ\tilde{k}\circ\varphi$ and $k_{\exp}:=\exp\circ\tilde{k}\circ\varphi$ and $k:=e^{-\alpha ^2/\sqrt{\beta}}k_pk_{\exp}$. The space $\mathcal{H}(k)$ can be realized as the pullback along the map $\Delta$ of the tensor Hilbert space $\mathcal{H}(\exp )\otimes\mathcal{H}(p)$. Indeed, $k$ induces the RKHS $\mathcal{H}(k)=\lbrace f:\R _+\rightarrow\R\vert f(x)=g(x)h(x), g\in\mathcal{H}(\exp ), h\in\mathcal{H}(p)\rbrace$ with induced norm $\Vert f\Vert _{\mathcal{H}(k)}=\min\lbrace\Vert g\Vert _{\mathcal{H}(\exp )}\Vert h\Vert _{\mathcal{H}(p)}, f=gh, g\in\mathcal{H}(\exp ), h\in\mathcal{H}(p)\rbrace$. This follows from \Cref{lem:rkhs_general} and \cite[Theorem 5.16]{paulsen_rkhs}.
	\end{enumerate}
\end{remark}
\begin{proposition}\label{prop:4_fully consistent_ker}
	Let $p_i(t)=a_{d,i}t^d+...+a_{1,i}t+a_{0,i}$ with $a_{k,i}\geq 0$ for $k=0,\dots ,d$ and $i=1,\dots ,d+1$ and let $\alpha _i,\beta _i\in\R$, $\alpha _i\geq 0,\beta _i >0$ for $i=1,\dots ,d+1$. Set $k_i(x,y)=p_i((\sqrt{\beta _i}x-\alpha _i/\sqrt{\beta _i})(\sqrt{\beta _i}y-\alpha _i/\sqrt{\beta _i}))e^{\beta _ixy-\alpha _i (x+y)}$ and define $k(x,y)=\sum _{i=1}^{d+1}k_i(x,y)$. Then $k$ is a fully consistent kernel in the sense of \Cref{def:fully consistent_kernel} and gives rise to the RKHS given by the direct sum $\mathcal{H}(k)=\mathcal{H}(k_1)\oplus\dots\oplus\mathcal{H}(k_{d+1})$, with norm
	\begin{equation}\label{eq:4_fully consistent_ker_norm}
		\Vert f\Vert ^2_{\mathcal{H}(k)}=\emph{min}\left\lbrace\sum _{i=1}^{d+1}\Vert f_i\Vert ^2_{\mathcal{H}(k_i)}:f=\sum _{i=1}^{d+1}f_i, f_i\in\mathcal{H}(k_i), i=1,\dots ,d+1\right\rbrace ,
	\end{equation}
	where the spaces $\mathcal{H}(k_1),\dots ,\mathcal{H}(k_{d+1})$ are defined as in \Cref{prop:4_rkhs}.
\end{proposition}
\section{Model calibration}\label{sec:calibration}
In the following section, we will calibrate our model to real market data. In particular, we will perform a two-step numerical procedure: 
\begin{itemize}
	\item[(1)] In the first step of our procedure, we will take observed (coupon) bond contracts and use the Representer Theorem to fit a discount curve from an admissible kernel space. Thus, for each day, we will obtain a discount curve as a function of the tenor, allowing us to extract time-series data needed to calibrate the underlying stochastic process for our risk-neutral model. Since the Representer Theorem implies that the inferred curve lies in a kernel subspace of dimension depending on the number of observed tenors, this yields a very high-dimensional model.
	\item[(2)] In the second step, we will fit a simple $d$-dimensional stochastic model to the inferred kernel-based curve from the first step, where $d$ will be much lower than the dimension of the implied kernel subspace. In fact, e.g. \cite[Section 3.4]{filipovic} and \cite{filipovic_kr} suggest that 3 to 4 factors already explain more than $99\%$ of the variance in the model, which may serve as a basis for our choice of dimension.
\end{itemize}
We will conduct our analysis using the CRSP dataset of US Treasury bonds\footnote{Dataset used: CRSP Treasuries (Annual) \copyright 2024 Center for Research in Security Prices, LLC (CRSP) \url{https://wrds-www.wharton.upenn.edu/data-dictionary/crsp_a_treasuries/}}. The data were cleaned and preprocessed using the same procedure as in \cite{filipovic3}. For our task, we will use data collected over one year, covering over $252$ trading days from January 1, 2021, to December 31, 2021. In \Cref{fig:4_bond_prices}, we provide the bond price data and the implied yield to maturities realized on the market on  December 31. 

\begin{figure}[h!]
		\includegraphics[width=13cm, height=5.5cm]{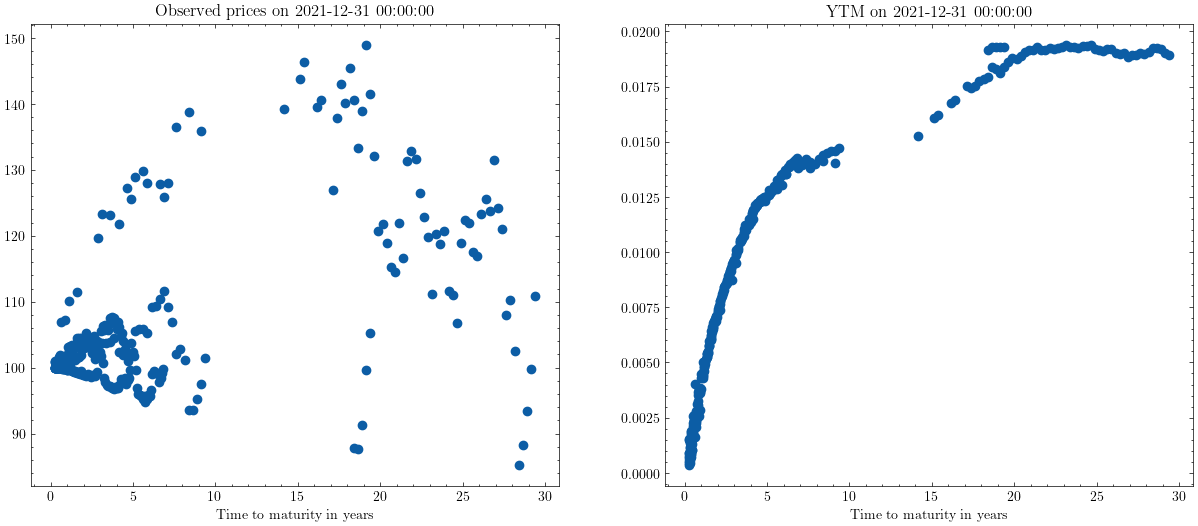}
		\caption{Bond prices and implied yield on 31st of December, 2021}\label{fig:4_bond_prices}
\end{figure}
The contracts offered on the market are coupon bonds with prices quoted on each trading day. Any coupon bond can be written as a linear combination of zero-coupon bonds with different times to maturity multiplied with respective coupons. That is, for any coupon bond $P$, we have
\begin{equation*}
	P = \sum _{i=1}^NC_ih(x_i),
\end{equation*}
where $\{x_1,\dots ,x_N\}$ is a collection of tenors, $C_i$ denotes the cashflow at maturity $x_i$ and $h$ denotes the price of the zero-coupon bond with time to maturity $x_i$. 

We aim to calibrate our model of the term structure of zero-coupon bonds to reproduce the bond prices $P$ observed on the market. Therefore, on any given trading day, we extract the cashflow matrices for a vector of observed coupon bond contracts and use the underlying zero-coupon curve as our variable of interest. In \Cref{fig:4_cashflow}, we depict the cash flow matrix extracted from the observed bond prices on the 31st of December.

\begin{figure}[h!]
		\includegraphics[width=13cm, height=7cm]{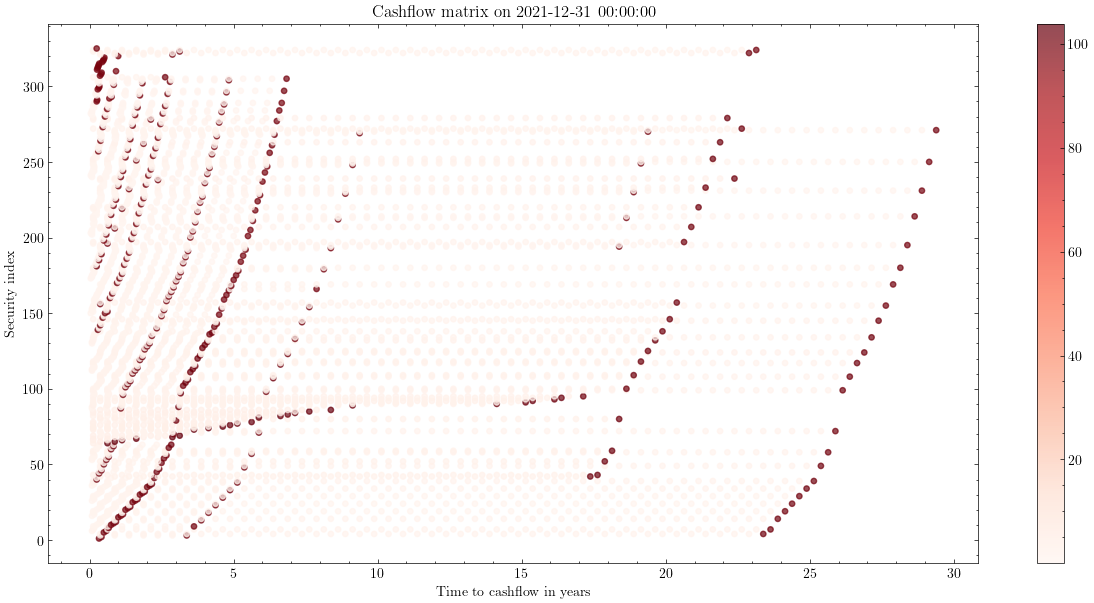}
		\caption{Cashflow matrix extracted from coupon bond on the 31st of December, 2021}\label{fig:4_cashflow}
\end{figure}

Thus, the preprocessed dataset is composed of vectors of zero-bond prices on any given observation day and corresponding cashflow matrices. Next, we continue with describing our procedure.
\subsection{First step optimisation}
Given a dataset with zero-bond prices and cashflow matrices, our aim is to fit a zero-coupon bond curve $h$, which will minimize the pricing error. To this end, we shall fix an appropriate curve space from which to draw our curve. Since we are interested in models fulfilling no-arbitrage, we will use a fully consistent kernel $k$ and fix an RKHS $\mathcal H(k)$. Our goal will now be to formulate and solve a suitable optimization problem. To this end, let $M\geq 0$ denote the number of contracts and $N\geq 0$ denote the number of available different tenors and consider on any given day the vector of quoted coupon bond prices $(P_1,\dots ,P_M)$ with corresponding cash flow matrix $C = (C_{ij})\in\R^{M\times N}$. Let $H$ be the discount and curve and let $C_i$ denote the $i$-th row of $C$ and $h:=1-H$ be the zero-coupon bond curve. Consider the cost functional
\begin{equation}\label{eq:4_cost_func}
	\mathcal{J}(h):=\sum _{i=1}^Mw_i\left(P_i-C_i((h(x_1),\dots ,h(x_N))^{\top}\right) ^2+\lambda\Vert h\Vert _{\mathcal H(k)}^2
\end{equation}
for $0< w_i\leq\infty$ and $\lambda >0$. We observe that $\mathcal J$ has two components: a weighted square-loss function of $h$ against the observed prices and an additional penalty term with coefficient $\lambda$ given by the norm in the RKHS. Indeed, the latter term controls the derivatives of the function $h$ and thus plays the role of a shape penalty term. Thus, we aim to faithfully reproduce prices observed in the market while penalizing functions that do not behave ``nicely'' enough. However, the minimization of the functional $\mathcal J$ is an infinite-dimensional regression problem, which is not tractable numerically unless one fixes an appropriate parametric curve family to minimize over. To approach this problem, we will formally introduce our main tool, the Representer Theorem (cf. \cite[Theorem 8.7.]{paulsen_rkhs}), which is the main justification for why the theory of RKHS is very useful for optimization:
\begin{theorem}\label{thm:4_representer}
	Let $\mathcal X$ be a set and let $k$ be a kernel on $\mathcal X$ with induced RKHS $\mathcal H(k)$. Let $W:\R\rightarrow\R$ be a monotonically increasing function and $\mathcal L:\R^n\rightarrow\R$ be continuous. Consider the cost functional 
	\begin{equation}\label{eq:4_representer1}
		\mathcal J(f):=\mathcal L(f(x_1),\dots ,f(x_n))+W(\Vert f\Vert _{\mathcal H(k)})
	\end{equation}
	for $\{x_1,\dots ,x_n\}\in \mathcal X$. If $f^*$ is a function such that $\mathcal J(f^*)=\inf _{f\in\mathcal H(k)}\mathcal J(f)$,
	then $f^*$ lies in the span of the functions $k_{x_1},\dots ,k_{x_n}$.
\end{theorem}
Indeed, by using the Representer Theorem, one may reduce the infinite-dimensional problem of minimizing \eqref{eq:4_cost_func} to a finite-dimensional ridge regression over the coefficients in the linear representation of the minimizer within the span of our kernels. Thus, we have the following
\begin{proposition}\label{prop:4_minimisation}
	Let $M,N\geq 0$ and $C\in\R ^{M\times N}$. Denote by $C_i:=(C_{i1},\dots ,C_{i,N})^{\top}$ the $i$-th row vector of $C$ and let $0<w_i\leq\infty$ for $i=1,\dots ,M$. Furthermore, define the index sets $\mathcal I_1:=\lbrace 1\geq i\geq M: w_i=\infty\rbrace$ and $\mathcal I_0:=\lbrace 1,\dots ,M\rbrace\backslash\mathcal I_1$. Consider the minimization problem
	\begin{equation}\label{eq:4_minimisation}
		\min _{h\in\mathcal H(k)}\left\lbrace \sum _{i\in\mathcal I_0}^Mw_i\left( P_i-C_i(h(x_1),\dots ,h(x_N))^{\top}\right) ^2+\lambda \Vert h\Vert ^2_{\mathcal H(k)}\right\rbrace .
	\end{equation}
	Let $K_{ij}=k(x_i,x_j)$ denote the kernel matrix induced by the reproducing kernel $k$ and $\Lambda :=\emph{diag}(\lambda /w_1,\dots ,\lambda /w_M)$, where we define $\lambda /\infty :=0$ and assume that either $\mathcal I_1=\emptyset$ or that $C_{\mathcal I_1}KC_{\mathcal I_1}^{\top}$ is invertible. Then the matrix $CKC^{\top}+\Lambda$ is invertible and there exists a unique solution $\hat h$ to \eqref{eq:4_minimisation} given by 
	\begin{equation}\label{eq:4_minimisation1}
		\hat h=\sum _{i=1}^M\alpha _ik(\cdot, x_i),
	\end{equation}
	where $\alpha=(\alpha _1,\dots ,\alpha _M)^{\top}$ is given by
	\begin{equation*}
		\alpha = C^{\top}(CKC^{\top}+\Lambda)^{-1}P
	\end{equation*}
\end{proposition}
\begin{proof}
	This follows immediately from \Cref{thm:4_representer} and \cite[Theorem A.1]{filipovic3}.
\end{proof}
Note that to satisfy the terminal bond payout condition $h(0)=1$, we add a soft constraint by introducing a synthetic cashflow of $1$ at maturity $0$ into the dataset for training. 

For the purposes of the numerical analysis, we will make use of the kernel
\begin{equation*}
	k_{\exp}(x,y):=e^{\beta xy-\alpha (x+y)}.
\end{equation*}
The corresponding RKHS $\mathcal H(\exp)$ we will use for our minimization is thus the Segal-Bargmann space. To begin our optimization procedure, we want to find kernel parameters $\alpha$ and $\beta$ and the ridge parameter $\lambda$, which will facilitate a good fitting. In order to find optimal parameters, a cross-validation procedure is used, which yields the following optimal set of estimates:
\begin{equation*}
	\{\alpha ,\beta ,\lambda\} = \{ 0.2, 0.04, 0.001 \}.
\end{equation*}

\Cref{fig:fit_price_full_model} presents a showcase by providing the estimation results for the day 31st of December.  
\begin{figure}[h!]
		\includegraphics[width=13cm, height=5.5cm]{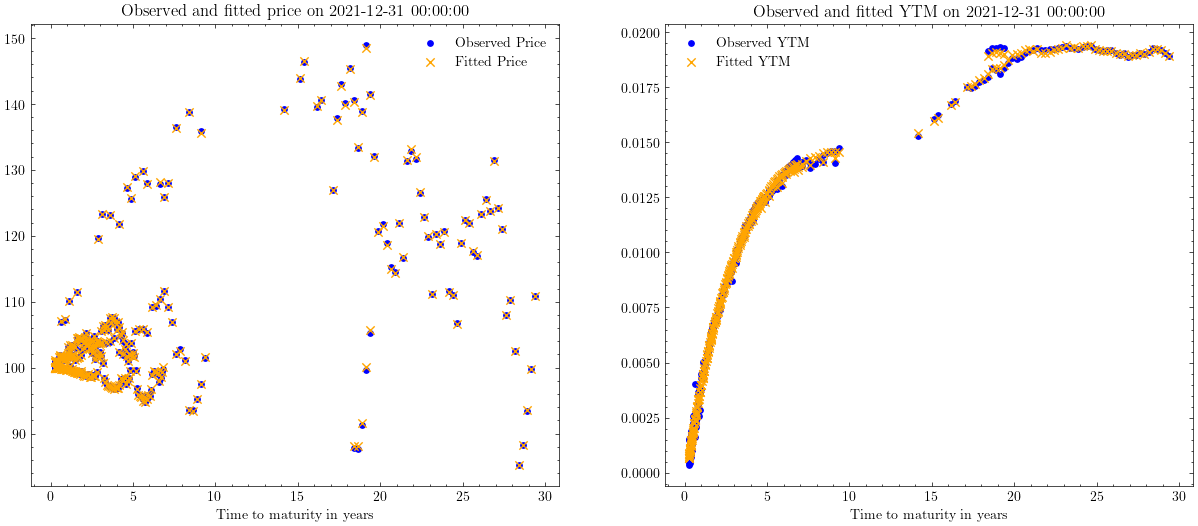}
		\caption{Observed data against fitted data on the 31st of December, 2021}\label{fig:fit_price_full_model}
\end{figure}

\begin{figure}[h!]\label{fig:full_model_curves}
		\includegraphics[width=13cm, height=5.5cm]{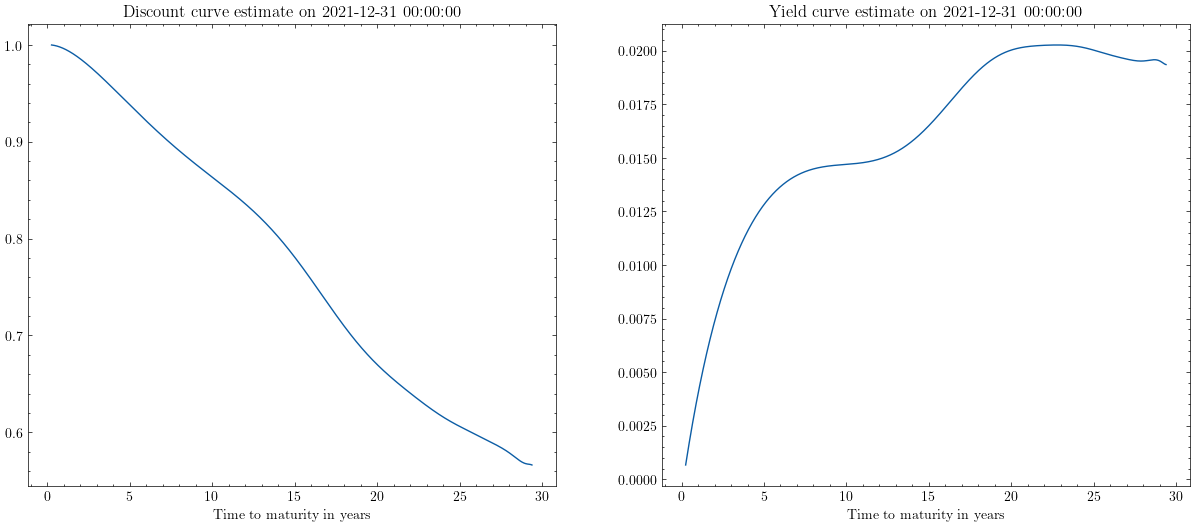}
		\caption{Implied zero-coupon price curve and yield curve on the 31st of December, 2021}
\end{figure}
We note that the successful fit results observed in \Cref{fig:fit_price_full_model} are not exclusive to the day chosen. Indeed, we observe that the results seem to carry over similarly for all of the trading days in the dataset with the average mean-squared error across all contracts on a given observation day around $0.000185$, that is $\approx 1.85$ basis points of the average yield.

\begin{figure}[h!]\label{fig:full_model_mse}
		\includegraphics[width=13cm, height=6.5cm]{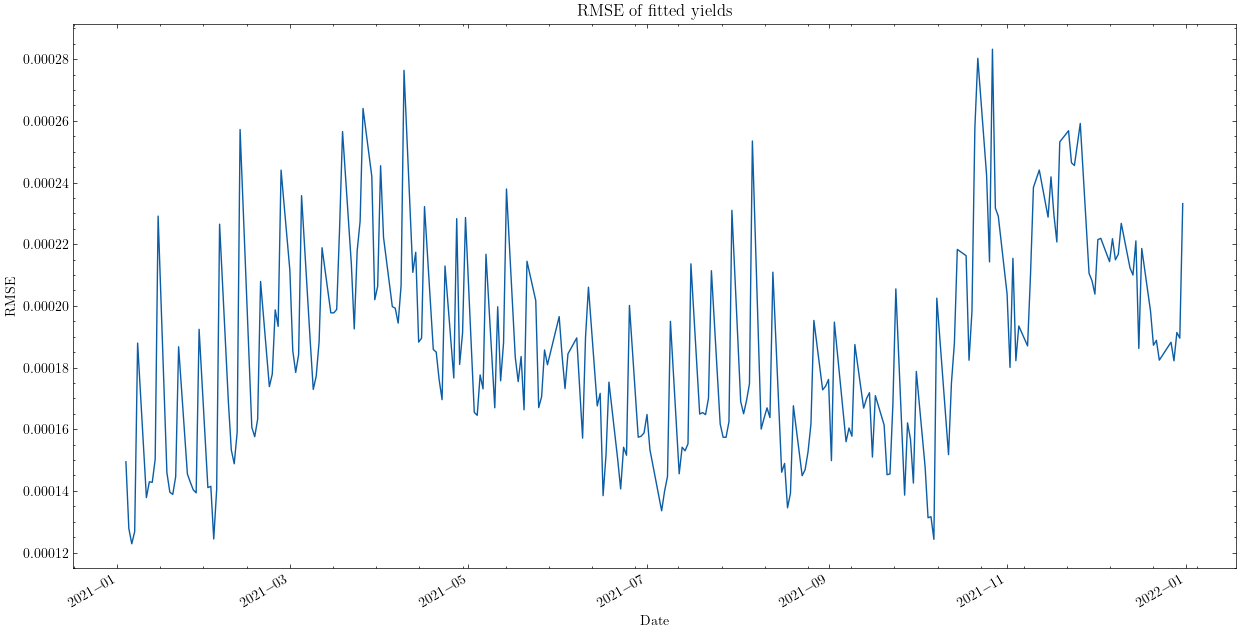}
		\caption{root mean square error of fitted yields across all contracts on each trading day in the observed time window}
\end{figure}
We also perform a sensitivity analysis, where we take the parameter set for the best fit and calibrate the model several times while varying the parameters in a range of $20\%$ of the original value to $500\%$. We capture the results in terms of root mean square errors of the yields and norm in the RKHS in heat maps in \Cref{fig:heat_map}. One parameter was fixed in each of the plots, while the remaining two varied. We notice that the model behaves relatively robustly with respect to the ridge parameter $\lambda$. For the kernel parameters $\alpha$ and $\beta$, the results show a very high sensitivity, particularly in the case of $\beta$. This comes as little surprise, as $\beta$ contributes to the exponent in a multiplicative way. Hence, small discrepancies, in particular towards larger values, lead to rapidly growing curves. We note that while reducing the value of the parameters $\alpha$ and $\beta$ often leads to worse results, this can be done jointly to obtain a fit that seems to be of a similar quality to our best result.   
\begin{figure}[h!]
	\includegraphics[width=14cm, height=9cm]{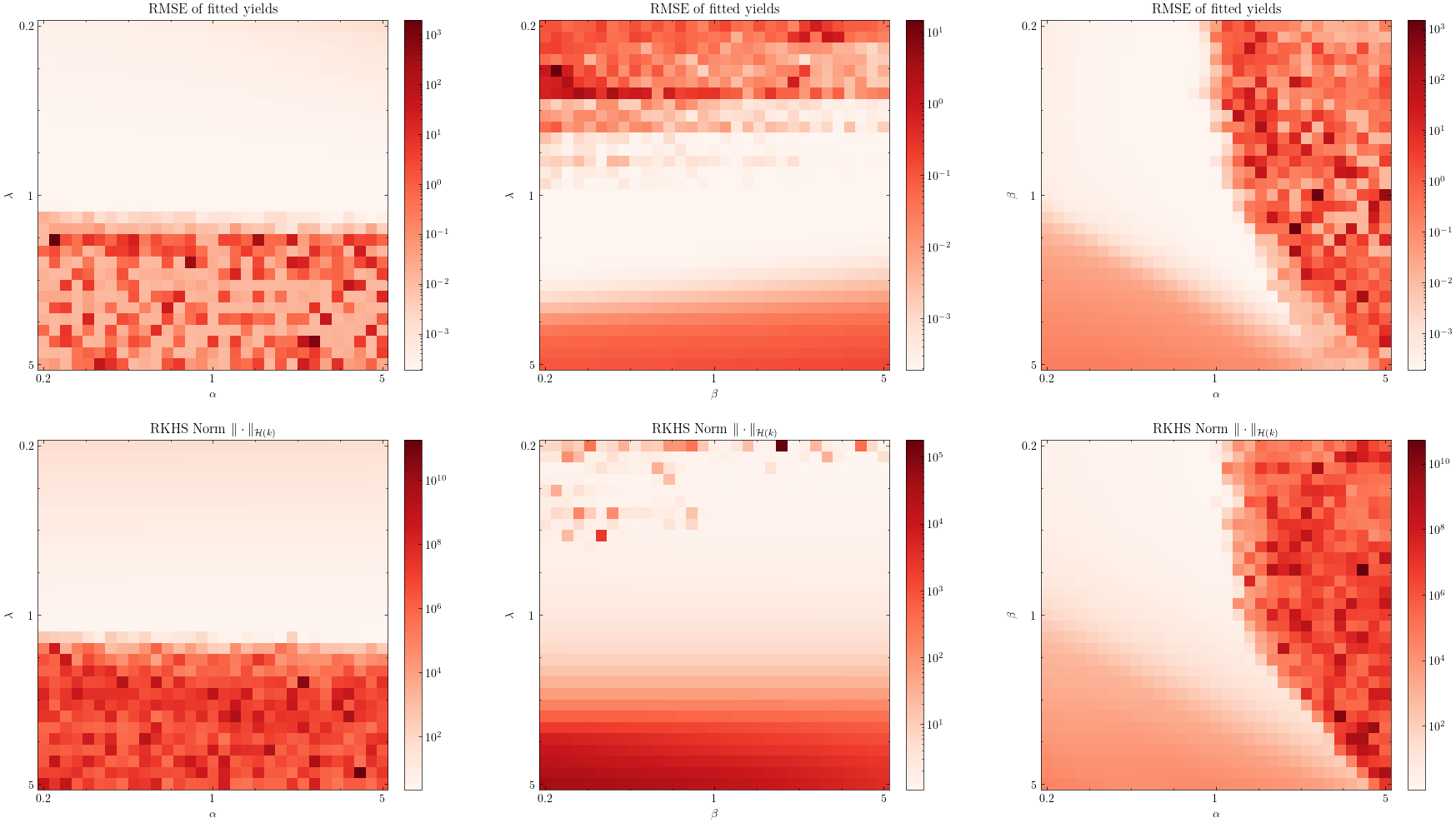}
		\caption{Sensitivity of the model with respect to change in parameters. In each plot, one of the parameters is kept fixed.}\label{fig:heat_map}
\end{figure}
\subsection{Second step optimisation}
We consider the results we obtained in the first step optimisation. For each observation day, we fitted a bond-price curve to the observed prices in the market using reproducing kernels as a regression basis. Using the Representer \Cref{thm:4_representer}, this yields at time $t$ a function of time to maturity which is a linear combination of kernels, that is
\begin{equation*}
	\hat h_t=\sum _{x\in\mathcal X_t}c_{x,t}k_x,
\end{equation*}
where $\mathcal X_t$ denotes the collection of tenors of the zero-coupon bonds available on day $t$, that is $\mathcal X_t:=\{ x^t_1,\dots ,x^t_{M_t}\}$ for some $M_t\in\mathbb N$ for all $t\geq 0 $. For our dataset, $M_t\approx 300$, typically. While pleasing from a numerical and fitting perspective, we want to compare with our implied stochastic model. Since $k$ is a fully consistent kernel, $\hat h_t$ generates an admissible term structure model. Indeed, we may take $h_t$ as is and define the model
\begin{equation*}
	\hat H_t(x)=\langle C_t,K(x)\rangle,
\end{equation*}
where $K=(k_{x_1},\dots ,k_{x_M})^{\top}$, where $\{x_1,\dots ,x_M\}=\bigcup _{t\geq 0}\mathcal X_t$ is the collection of distinct tenors available across all times $t$ and $C$ is the process defined as
\begin{equation*}
	\begin{aligned}
		C_{t,i}=\begin{cases}
			c_{x_i, t},\quad&\text{if }x_i\in\mathcal X_t,\\
			0,&\text{else.}
		\end{cases}
	\end{aligned}
\end{equation*}
Indeed, with this specification, we obtain a consistent $M$-dimensional model $H:=1-\hat H$ of the type specified in \Cref{cor:4_diagonalisable_m}, where $M\leq T\max _{t\geq 0}M_t$, $K\equiv g$ and $C_t$ is one realisation of the stochastic process $Z_t$, that is $C_t=Z_t(\omega _0)$ for some $\omega _0\in\Omega$. We will henceforth refer to it as the ``full'' model. This model satisfies
\begin{equation*}
	\hat H_t =\hat h_t\quad\text{for all }t\geq 0.
\end{equation*}
Since the longest time to maturity in the data set is approximately $30$ years and we consider tenors in day steps, we have $M\leq 30\times 365$. While consistent, the dimensionality is far from satisfactory. Indeed, classic results using PCA, suggest that $4$ dimensions explain more than $99\%$ of the variance in term structure HJM models (see, e.g. \cite[Section 3.4.]{filipovic}). For the second step optimisation, we will therefore aim to find a consistent model which will perform comparably well to the full model, but with a lower-dimensional specification. To this end, we will define an appropriate finite-dimensional subspace of the RKHS $\mathcal H(k)$ and minimise a loss functional with respect to the full model.

To begin, we define our model specification. Under the assumption of an affine discount model, by \Cref{prop:4_3}, for any $d$-dimensional consistent model there is some matrix $M\in\R ^{(d+1)\times (d+1)}$ such that $g(x)=(\mathbbm 1_{d+1}-e^{xM})e_0$. Let now $M\in\R ^{(d+1)\times (d+1)}$ and $\sigma :\R _+\times\Omega\rightarrow\R ^{d\times d}$ be a progressively measurable matrix process with a.s. integrable paths and let $H$ be affine no-arbitrage model induced by $(M, y_0,\sigma)$ in the sense of \Cref{prop:4_3}. Due to \Cref{cor:4_basis_change}, we may write
\begin{equation}\label{eq:4_base_model}
	H_t(x)=1-\langle e^{Jx}p,Z_t\rangle = 1-\sum _{i=0}^dZ_{t,i}q_i(x)e^{\lambda _ix},
\end{equation}
where $J$ is an appropriate Jordan block matrix with Eigenvalues $\{\lambda _0,\dots ,\lambda _d\}$, $q_i\in\text{Pol}_d(\R )$ for $i=0,\dots ,d$ and $Z$ is a stochastic process fulfilling the no-arbitrage quadratic drift condition. We therefore want to find a model $H$ of low dimension (henceforth referred to as the ``reduced'' model) such that the discrepancy with the full model is minimised. We proceed by defining appropriate subspaces for the implied constrained minimisation problem.

\begin{definition}\label{def:q_exp}
	Let $k$ be a reproducing kernel, $\mathcal{H}(k)$ be the RKHS induced by $k$, and let $d\in\mathbb N$. Define the space
	\begin{equation}\label{eq:4_q_exp1}
		E^d(k):=\set[\Bigg]{ \sum _{i=1}^d\eta _ik(\cdot ,y_i)\given \eta _i\in\R ,y_i\in\R _+,i=1,...,d }
	\end{equation}
\end{definition}

\begin{remark}
    Note that $E^d(k)$ is not a vector space, but rather a union of vector spaces. To see this, define
    \begin{equation*}
        \mathcal E^d(y_1,\dots ,y_d;k):=\set[\Bigg]{ \sum _{i=1}^d\eta _ik(\cdot ,y_i)\given \eta _i\in\R ,i=1,...,d }.
    \end{equation*}
    One can easily show that $\mathcal E^d(y_1,\dots y_d;k)$ is a vector space for any fixed $\{y_1,\dots, y_d\}\subset\R$. Then we have
    \begin{equation*}
        E^d(k)=\bigcup _{\{ y_1,\dots ,y_d\}\subset \R_+}\mathcal E^d(y_1,\dots ,y_d; k).
    \end{equation*}
\end{remark}
Let now $h_t:=1-H_t$ denote the zero-coupon bond price curve implied by the simpler model. In order to fit such a parsimonious model where $d\ll M$, we therefore need to minimise the functional
\begin{equation}\label{eq:4_minimiser_functional}
	\mathcal L:=\sum _{t=0}^T\Vert\hat h_t-h_t\Vert _{\mathcal H(k)}^2
\end{equation}
over the chosen admissible set $E^d(k)$, that is we need to solve the constrained optimisation problem
\begin{equation*}
	\min _{h:=\{h_t\} _{t\geq 0}\subseteq E^d(k)}\mathcal{L}(h).
\end{equation*}
We first provide an existence result for a minimiser on our admissible subset. Note that uniqueness of the minimiser is not necessarily maintained when restricting to a subset.
\begin{proposition}\label{prop:4_existence_min}
	Let $k$ be a fully consistent kernel of the form $k=k_p\cdot k_{\exp}$, $\mathcal{H}(k)$ be the corresponding RKHS and let $d\in\mathbb{N}$. Let $E^d(k)$ be defined as in \Cref{def:q_exp} and $h\in\mathcal{H}(k)$. Consider the functional
	\begin{equation}\label{eq:4_existence_min1}
		\mathcal L(g):=\Vert h-g\Vert _{\mathcal H(k)}^2
	\end{equation}
	Then the following statements hold true.
	\begin{enumerate}
		\item Assume $p(x,x) > 0$ for all $x\in\R _+$, then $\mathcal L$ attains its minimum over $E^d(k)$.
		\item There exist symmetric, positive semidefinite polynomials $q:\R _+\rightarrow\R$ and $r:\R _+\times\R _+\rightarrow\R$ with $p(x,y)=q(x)q(y)r(x,y)$ and $r(x,x)>0$ for all $x\in\R _+$, and $\mathcal L$ attains its minimum over $qE^d(k_r\cdot k_{\exp}):=\{qf: f\in E^d(k_r\cdot k_{\exp})\}$.
	\end{enumerate}
\end{proposition}
\begin{proof}
	See \Cref{app:b}.
\end{proof}
We shall now proceed with the optimisation step. For our purposes, we will fix the set $E^d$, that is, we will consider models of the form as in \Cref{eq:4_base_model} where $\text{deg}(q_i)=0$ for all $i=0,\dots ,d$. This corresponds to choosing a diagonalisable $M\in\R^{(d+1)\times (d+1)}$ such that the triplet $(M, y_0, \sigma )$ generates the model $H$. Since the set of diagonalisable matrices is dense in the set of matrices, the use of a simplified model of the type given in \Cref{cor:4_diagonalisable_m} is justified in the numerical calibration. 

\begin{proposition}\label{prop:4_sol_2nd_minimisation}
	Let $\hat h_t:=\sum _{i=1}^{M_t}c_{t,i}k_{x_i}$ and $h_t\in E^d(k)$ for $t=0,\dots ,T$, $d\in\mathbb N$. Consider the minimisation problem
	\begin{equation}\label{eq:4_min_problem}
		\min _{c_t}\Vert\hat h_t-h_t\Vert _{\mathcal H}.
	\end{equation}
	Define the matrices $K'_t\in\R ^{M_t\times (d+1)}$ and $K''\in\R ^{(d+1)\times (d+1)}$ with entries
	\begin{equation}\label{eq:4_def_matrices2}
		\begin{aligned}
			(K'_t)_{ij}&=e^{\lambda _jx_i}\qquad&\text{for }i\in\lbrace 1,\dots ,M_t\rbrace, j\in\lbrace 0,\dots ,d\rbrace,\\
			(K'')_{ij}&=\langle e^{\lambda _i\cdot},e^{\lambda _j\cdot}\rangle _{\mathcal H},\qquad&\text{for }i,j\in\lbrace 0,\dots ,d\rbrace,
		\end{aligned}
	\end{equation}
	and assume the matrix $K''$ is invertible. Then the solution vector $\hat a _t$ is given by
	\begin{equation*}
		\hat c _t =(K'')^{-1}(K'_t)^{\top}\eta_t\qquad\text{for }t=0,\dots ,T,
	\end{equation*}
\end{proposition}
\begin{proof}
	Consider the kernel $k=k_p\cdot k_{\exp}$ and the space $E^d(k_p\cdot k_{\exp})$. Let $\mathcal J_t:=\Vert\hat h_t-h_t\Vert ^2_{\mathcal H(k)}$ for $h_t\in E^d(k_p\cdot k_{\exp})$ for $t\in\mathbb N$. Then $h_t=\sum _{i=1}^dq\eta _{t,i} k'_{y_i}$ for $\eta _i\in\R$, $y_i\in\R _+$, $i=1,\dots ,d$, where $k':=r\cdot\exp$, and $q$ and $r$ are polynomials as specified in \Cref{lem:q_factorisation2}. Using the bilinearity of the inner product of the RKHS, we may expand the quadratic form 
	\begin{equation*}
		\begin{aligned}
			\Vert\hat h_t-h_t\Vert_{\mathcal H(k)}^2&=\langle\hat h_t-h_t,\hat h_t-h_t\rangle _{\mathcal H(k)}\\
                                    &=\langle\hat h_t,\hat h_t\rangle _{\mathcal H(k)}-2\langle\hat h_t,h_t\rangle _{\mathcal H(k)}+\langle h_t,h_t\rangle _{\mathcal H(k)}\\
								&=\sum _{i=1}^{M_t}\sum _{j=1}^{M_t}c_{t,i}c_{t,j}\langle k_{x_i},k_{x_j}\rangle _{\mathcal H(k)} -2\sum _{i=1}^{M_t}\sum _{j=0}^dc_{t,i}\eta_{t,j}\langle k_{x_i}, qk'_{y_j}\rangle _{\mathcal H(k)}\\
								&+\sum _{i=0}^d\sum _{j=0}^d\eta_{t,i}\eta_{t,j}\langle qk'_{y_i},qk'_{y_j}\rangle _{\mathcal H(k)}.
		\end{aligned}
	\end{equation*}
	Using \Cref{prop:4_isometry} and the fact that $k(x,y)=q(x)q(y)k'(x,y)$, we note that 
	\begin{equation*}
		\begin{aligned}
			\langle qk'_x,qk'_y\rangle _{\mathcal H(k)}&=\langle M_q(k'_x),M_q(k'_y)\rangle _{\mathcal H(k)}=\langle k'_x, k'_y\rangle _{\mathcal H(k')}=k'(x,y),\\
			\langle k_x,qk'_y\rangle _{\mathcal H(k)}&=\langle q(x)qk'_x,qk'_y\rangle _{\mathcal H(k)}=q(x)\langle M_q(k'_x),M_q(k'_y)\rangle _{\mathcal H(k)}=q(x)k'(x,y). 
		\end{aligned}
	\end{equation*}
	This yields
	\begin{equation*}
		\Vert\hat h_t-h_t\Vert_{\mathcal H(k)}^2=c_t^{\top}K_tc_t-2c_t^{\top}K'_t\eta_t+\eta_t^{\top}K''\eta_t,
	\end{equation*}
	where $K_t\in\R^{M_t\times M_t}$, $K'_t\in\R ^{M_t\times (d+1)}$, $K''\in\R^{(d+1)\times (d+1)}$ are matrices with entries
	\begin{equation*}
		\begin{aligned}
			(K_t)_{ij}&=k(x_i,x_j)\qquad&\text{for }i,j\in\lbrace 1,\dots ,M_t\rbrace,\\
			(K'_t)_{ij}&=q(x_i)k'(x_i,y_j)\qquad&\text{for }i\in\lbrace 1,\dots ,M_t\rbrace, j\in\lbrace 0,\dots ,d\rbrace,\\
			(K'')_{ij}&=k'(y_i,y_j),\qquad&\text{for }i,j\in\lbrace 0,\dots ,d\rbrace.
		\end{aligned}
	\end{equation*}
	We have the following first order condition
	\begin{equation*}
		\nabla _{c_t}\mathcal J_t=2K''c_t-2(\eta_t^{\top}K'_t)^{\top}=0. 
	\end{equation*}
	The assertion follows by solving for $c_t$.
\end{proof}
We provide here a slightly more general expression for the inner product of the RKHS with the exponential kernel. To be more precise, this pertains if one would like to extend the parametric families of functions $E^d$ to the case of polynomial coefficients. We remark here, that to the authors' knowledge in this case, the existence of a minimiser for the optimisation problem \eqref{eq:4_existence_min1} is not clear.
\begin{lemma}\label{lem:inner_product_specific}
	Let $k(x,y):=e^{\beta xy -\alpha (x+y)}$ and $p,q\in\text{Pol}(\R )$. Then
	\begin{equation*}
		\langle pe^{\lambda\cdot},qe^{\mu\cdot}\rangle _{\mathcal H(k)}=p(\partial _{\lambda})q(\partial _{\mu})e^{(\lambda -\alpha )(\mu +\alpha )/\beta}
	\end{equation*}
\end{lemma}
\begin{proof}
	Let $a:=(\lambda +\alpha )/\beta$ and $b:=(\mu +\alpha )/\beta$. Then
	\begin{equation*}
		\langle e^{\lambda\cdot},e^{\mu\cdot}\rangle _{\mathcal H(k)}=e^{(\lambda -\alpha )(\mu +\alpha )/\beta}=e^{-\alpha (a+b)}k(a,b).
	\end{equation*}
	We have for any $k,l\in\mathbb N$
	\begin{equation*}
		\langle \cdot ^ke^{\lambda\cdot},\cdot ^le^{\mu\cdot}\rangle _{\mathcal H(k)}=\langle\partial _{\lambda}^ke^{\lambda\cdot},\partial _{\mu}^le^{\mu\cdot}\rangle _{\mathcal H(k)}.
	\end{equation*}
	By \Cref{lem:dx_product} and bilinearity of the inner product,
	\begin{equation*}
		\langle\partial _{\lambda}^ke^{\lambda\cdot},\partial _{\mu}^le^{\mu\cdot}\rangle _{\mathcal H(k)}=\partial _{\lambda}^k\partial _{\mu}^l\langle e^{\lambda\cdot},e^{\mu\cdot}\rangle _{\mathcal H(k)}=\partial _{\lambda}^k\partial _{\mu}^le^{(\lambda -\alpha )(\mu +\alpha )/\beta}
	\end{equation*}
	Again, by bilinearity of the inner product, we have for general polynomials $p$ and $q$
	\begin{equation*}
		\langle pe^{\lambda\cdot},qe^{\mu\cdot}\rangle _{\mathcal H(k)}=p(\partial _{\lambda})q(\partial _{\mu})e^{(\lambda -\alpha )(\mu +\alpha )/\beta}
	\end{equation*}
	as asserted.
\end{proof}
We perform the fitting procedure for $d\in\{ 1,\dots ,30\}$ and capture some of the results. Firstly, we fix a generic trading day in the data set and observe how well the fit behaves with growing dimensionality of the reduced model.
\begin{figure}[h!]
		\includegraphics[width=13cm, height=6cm]{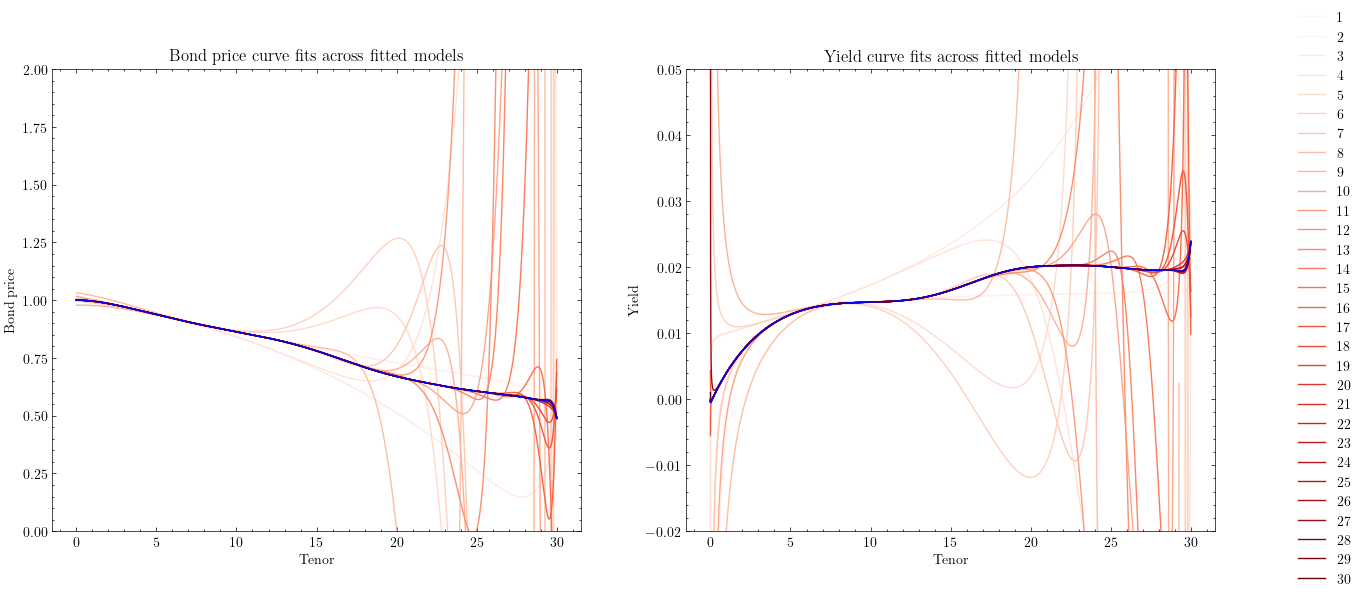}
		\caption{Curves produced by reduced models with full model (blue) on the 31st of December, 2021}\label{fig:curve_fit}
\end{figure}
In \Cref{fig:curve_fit} we see that the curves become indistinguishable to the eye starting around dimension $d=20$. In particular, the shape of the curve seems to be captured very well in the minimisation with the RKHS norm. This suggests two things: firstly, the RKHS seems to be well-suited for performing fitting procedures with a shape penalisation, and secondly, that the full model has a very high amount of redundancy. For the lower dimensional reduced models, we see an acceptable fit for shorter time to maturities, with a drop in performance towards the long-end. This may be due to the high density of low-term contracts available in the data set and a sparsity for medium-term contracts with fewer contacts towards the long end. Additionally, small deviations from the nominal price of $1$ at time to maturity $0$ results in large errors for the yield. This may suggest that a simple soft constraint may not be sufficient for model regression. Extrapolation with low-dimensional reduced models also seems to still have room for improvement. In order to quantify the quality of fit better, we present the fitting errors with respect to the contract prices available in the data set. Since we are optimising with respect to the exponents, which correspond to the time to maturity in the full model, a low amount of contracts with great times to maturity implies a lower number of exponents which control the rapid growth of the bond curve, hence a bigger error when extrapolating for time to maturity.
\begin{figure}[h!]
		\includegraphics[width=13cm, height=7cm]{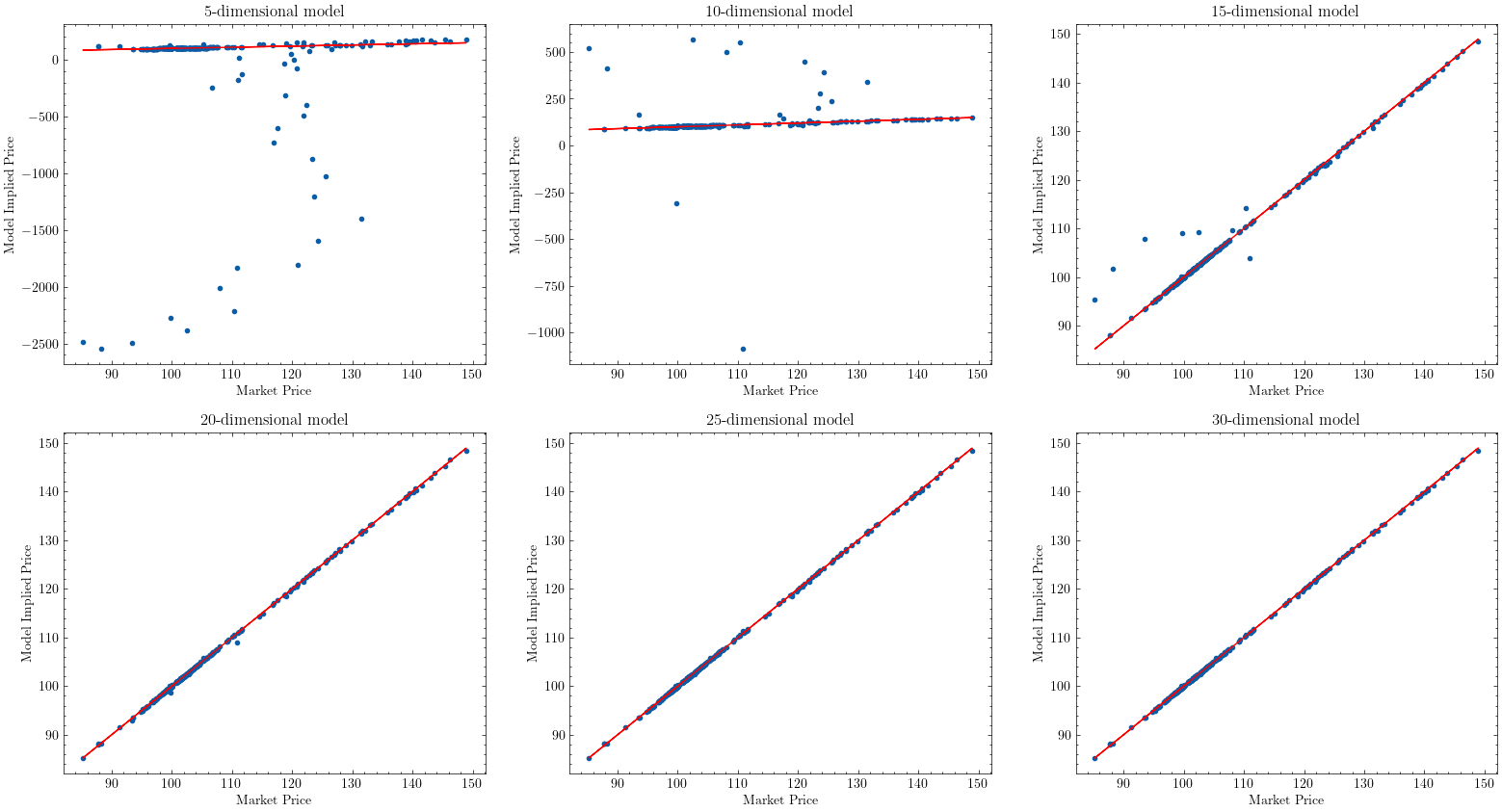}
		\caption{Price fits on reduced models on the 31st of December, 2021}\label{fig:price_fits1}
\end{figure}
In \Cref{fig:price_fits1} we proceed in increments of $5$ in dimensionality of the reduced model to observe how well the real market bond prices are approximated using our model. The prices in data set are plotted in a scatter plot against the model implied prices, along with the identity line for ease of comparison. Note here that we are comparing market prices with the reduced model prices, not the prices implied by the full model. It can be seen that the $20$-dimensional model and onwards already have negligible errors. For the lower-dimensional models, we observe discrepancies in the prices, which arise mostly due to the error with respect to long time to maturity contracts. It appears that towards the short-end, the model can capture prices nicely, but errors grow exponentially with longer tenors. We present the relative pricing errors of the reduced models as a function of time to maturity to emphasise this behaviour. The output has been truncated, with errors growing large for the low-dimensional models. 

\begin{figure}[h!]
		\includegraphics[width=13cm, height=7cm]{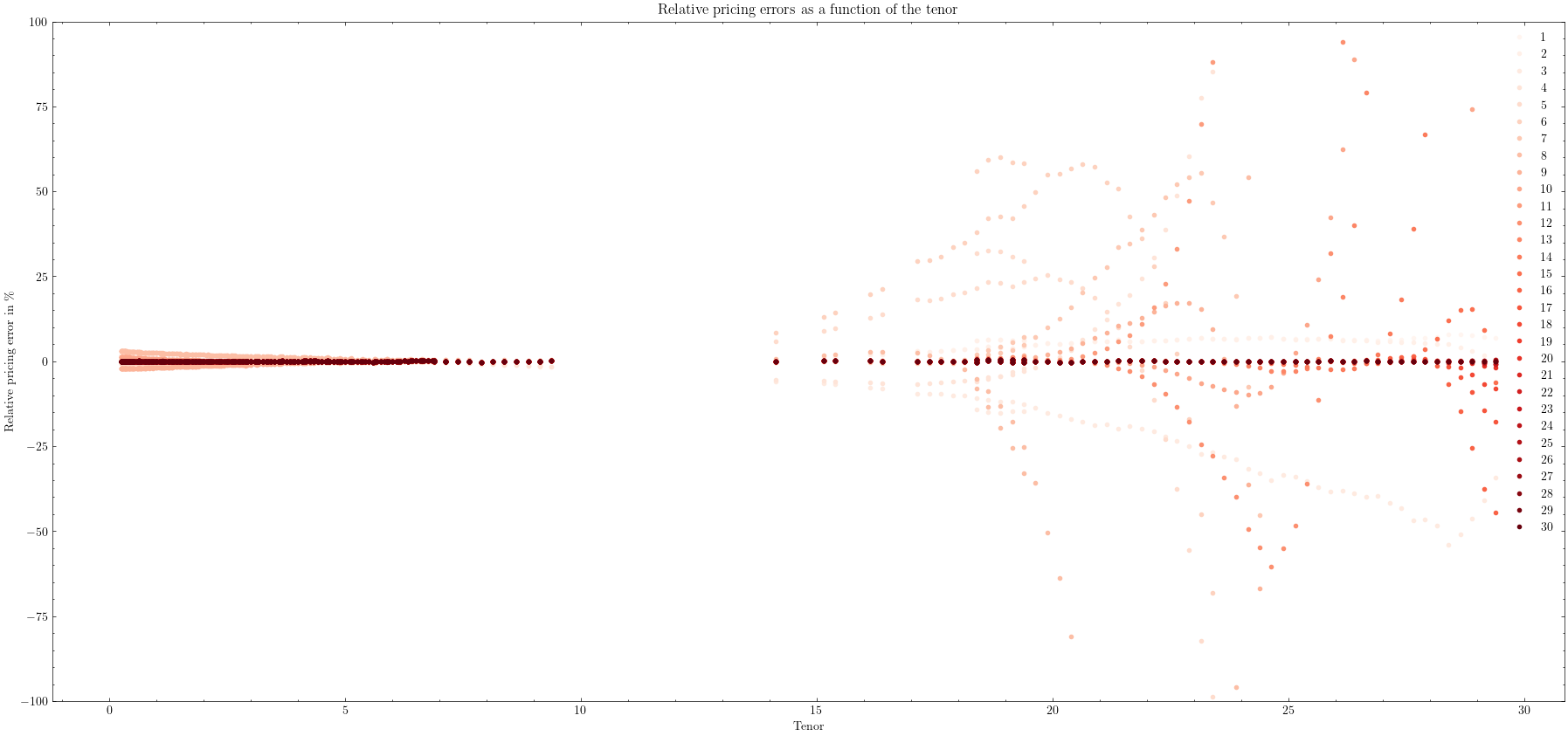}
		\caption{Relative pricing errors of reduced models (truncated at $100\%$) on the 31st of December, 2021}\label{fig:absolute_errs}
\end{figure}
\Cref{fig:absolute_errs} shows relative pricing errors for contracts with an average nominal value of $P\approx 100$. Errors for the short-end contracts appear well-behaved, even for the low-dimensional reduced models, as opposed to the long-end contracts. We observe here again the higher density in the region of short-term maturities which may be a contributing factor to the performance of the model. It can also be seen that the higher-dimensional reduced models produce an almost perfect fit to the market prices for all observed maturities. Performing a second step optimisation for higher dimensions of reduced models seems to be unnecessary. Finally, we capture the root mean square errors of fitted yields across all contracts on a given trading as a time series evolution on all days in our dataset for the reduced models. 
\begin{figure}[h!]
		\includegraphics[width=13cm, height=7cm]{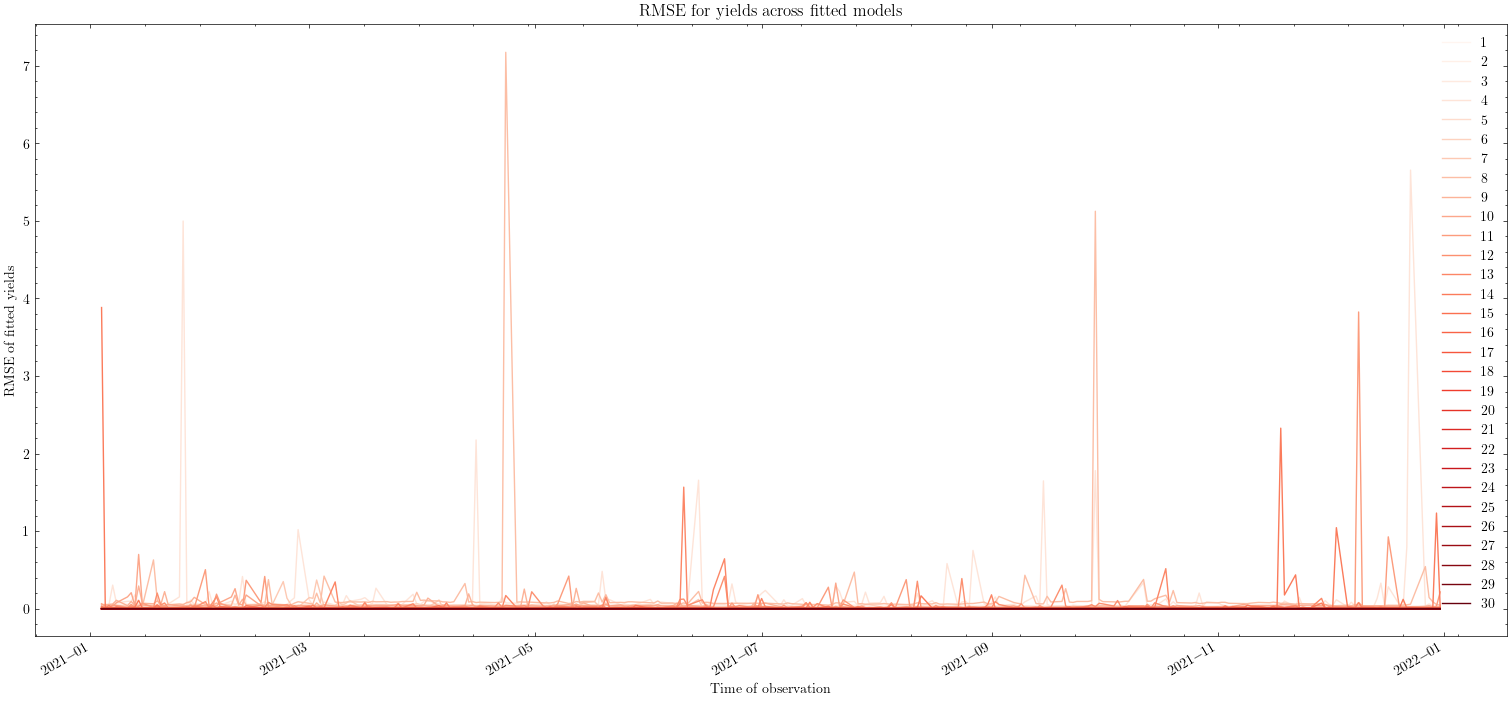}
		\caption{Root mean square errors of fitted yields on each trading day in the observed time window}\label{fig:mse}
\end{figure}
We again observe that the higher-dimensional reduced models produce a very good fit while the lower-dimensional models have some errors. Again, these occur mainly in the long-term contracts. The long-end of the curve seems to require higher dimensionality to be captured well. To better see the behaviour of errors, we compare the root mean square error of the higher-dimensional models, where the error is the lowest to the full model.
\begin{figure}[h!]
		\includegraphics[width=13cm, height=7cm]{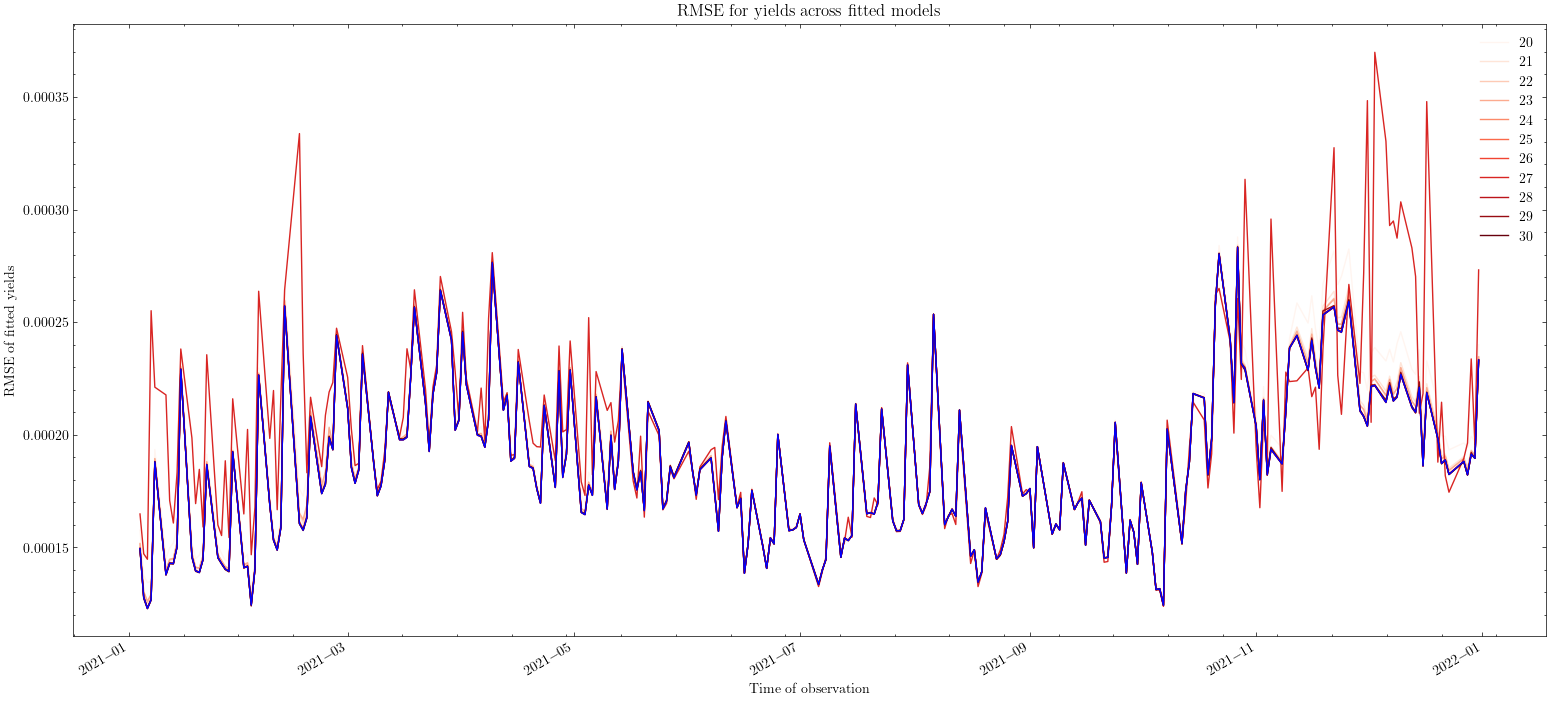}
		\caption{Comparison of the RMSE of the best performing models to the full model (blue)}\label{fig:mse2}
\end{figure}

To conclude our analysis, we consider the stochastic driving factors inherent to our model. Indeed, we have already observed a high amount of redundancy in going from the full model to the reduced model. Next, we want to see if we can still reduce the models to a ``minimal'' driving model by identifying correlation in the factors. We begin by noting that since the reduced models obtained in the second step optimisation are of the form \eqref{eq:4_base_model}, we may use the coefficients observed during the fit as a particular realisation of the underlying stochastic driving process, as in the case of the full model, that is we have
\begin{equation*}
	a_{t,i}=Z_{t,i}(\omega _0)\quad\text{for some }\omega _0\in\Omega.
\end{equation*}
Furthermore, we obtain as a result of the optimisation procedure the exponents of the exponential summands, which correspond exactly to the Eigenvalues of the matrix $M$. Indeed, with the model specification from \Cref{cor:4_diagonalisable_m}, we obtain a calibration of the curve components of the model, as well as the drift of the stochastic process due to the no-arbitrage quadratic drift condition. We thus do not have to estimate the drift, which can often be quite complex, but are left with estimation of the diffusion matrix $\sigma$, a far more tractable problem. For our purposes, we simply estimate the covariance matrices of the time series $a_t$ obtained from the fit and use these as a bootstrap for a constant diffusion matrix $\sigma$ for our simulations. 
\begin{figure}[h!]
		\includegraphics[width=13cm, height=4cm]{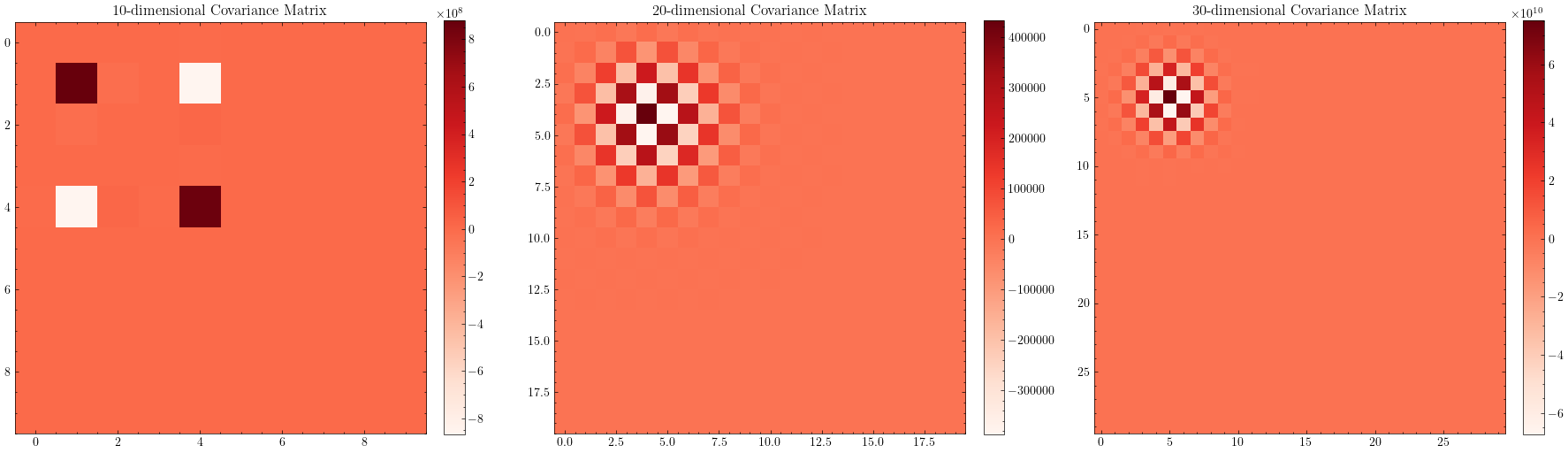}
		\caption{Estimated covariance matrices of extracted stochastic coefficients.}\label{fig:covs}
\end{figure}
From the heat maps in \Cref{fig:covs} we confirm our suspicion that some factors show strong correlation, even for the lower-dimensional models. This behaviour of course is amplified for the high-dimensional model. This suggests that there are still factors which could possibly be neglected without losing too much performance. However, adding our results for the fitting errors for the low-dimensional models, we see that our second-step optimisation appears to be unable to recognise which factors are negligible. It seems some of the correlated factors perform the task of a control factor for the higher exponents in the model. When left out, this leads to the high errors in the long-end of the curve.

To simulate our stochastic process, we use the drift condition and initialise $\sigma$  with the extracted covariance matrices. The resulting paths are captured in \Cref{fig:simulations}. Here, we notice that the dynamics of the simulated paths look reasonable when compared to the observed time series of the coefficients. This suggests that the process with no-arbitrage dynamics, even with a constant diffusion factor, seem to reflect the data observed on the market reasonably well. 
\begin{figure}[h!]
	\includegraphics[width=12.5cm, height=7cm]{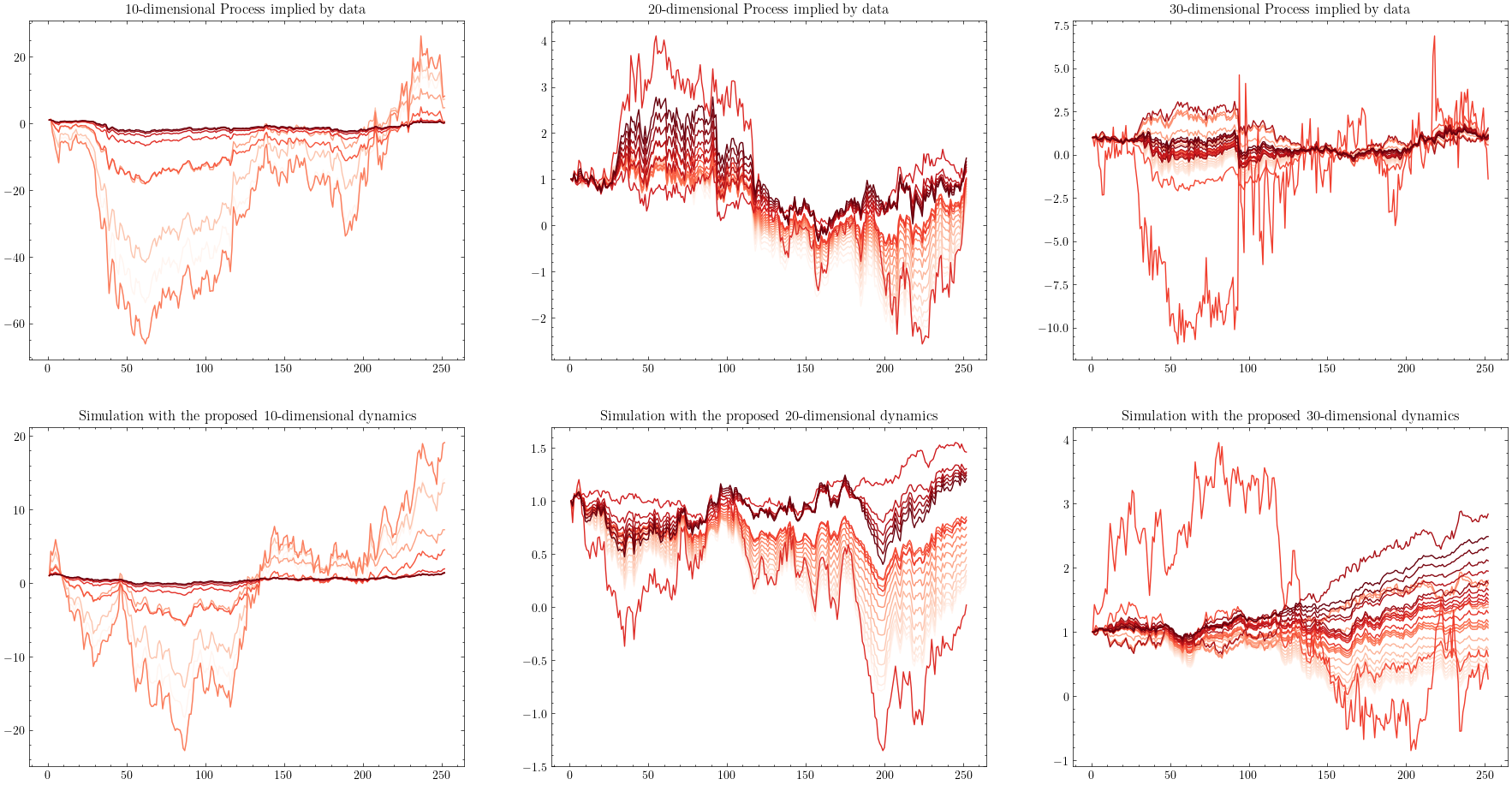}
	\caption{Extracted paths vs. simulated paths of stochastic processes calibrated using the covariance data normalised by their starting values.}\label{fig:simulations}
\end{figure}
Due to the simplicity of the model, it is now easy to simulate the entire term structure for arbitrary time frames. However, due to the observed bad extrapolation properties, simulating curves for tenors which the model has not been trained on does not yield reasonable results. We provide our results for simulations across a time frame of $252$ days, the same as in the learning set.
\begin{figure}[h!]
	\includegraphics[width=13cm, height=10cm]{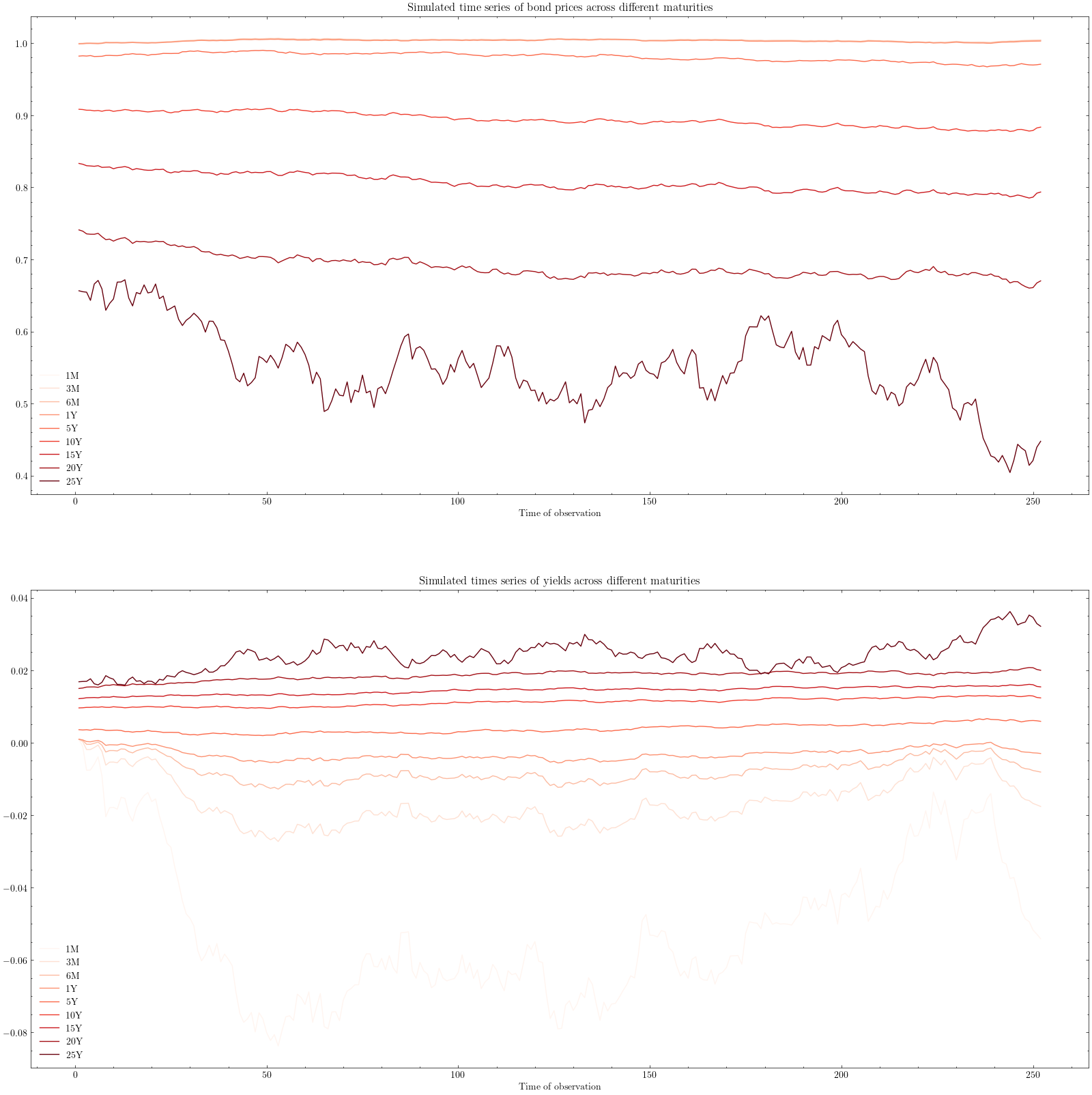}
	\caption{Simulation of time series of bond prices and yields for contracts with several different tenors.}\label{fig:simulations2}
\end{figure}

We observe that for longer tenors, the time series becomes more volatile. This is mainly due to fall in performance of the model for long tenors. Indeed, we notice a deterioration in the model starting at tenors around $25$ years. This can be attributed to the sparsity of long tenor contracts available in the training data. In absence of such data, one solution might be to include synthetic contracts with price points around the same as existing contracts. Furthermore, adding synthetic contracts with tenors far exceeding those observed in the data can serve as a soft constraint in the fitting process which improves extrapolation. Towards the short end, implementing a hard constraint might lead to an improvement on the results with respect to the yield curve, as small deviations can result in large errors.
\begin{figure}[h!]
	\includegraphics[width=13cm, height=6cm]{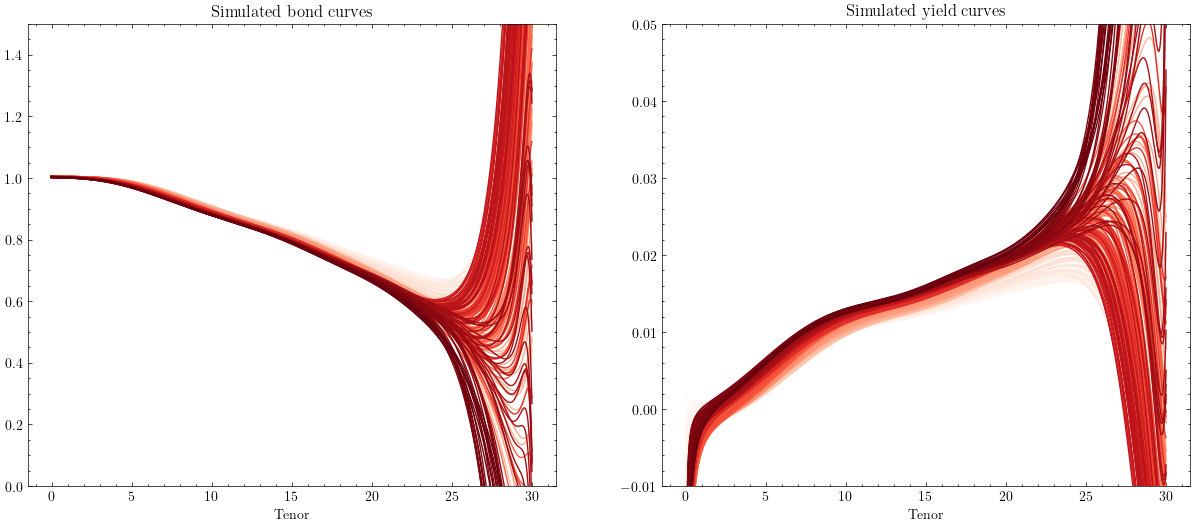}
	\caption{Simulation of bond price curves for every observation point in the time frame.}\label{fig:simulations3}
\end{figure}

\begin{remark}
    One possible way to facilitate better extrapolative properties in case of the exponential kernel is by restricting the choice of model parameters $\alpha$ and $\beta$ for the kernel regression. Observe in the full model specification the curve $\hat h_t$ is a linear combination of exponential functions, that is
    \begin{equation*}
        \hat h_t =\sum _{y\in\mathcal X_t}c_{y,t}e^{\beta y(\cdot)-\alpha (y+\cdot)}.
    \end{equation*}
    Thus, there is a constant $C\in\R$ with $\lvert h_t(x)\rvert\leq Ce^{y_{\max}x}$ for all $x,t\geq 0$, where $y_{\max}:=\max _t\max \mathcal X_t$.

    Thus, we may ensure the function $\hat h_t$ stays bounded for all $t\geq 0$ if we can control the exponents of the kernel. We observe that as $x$ grows, the exponential kernel stays bounded for any fixed $y\geq 0$ if and only if $\beta xy-\alpha (x+y)<0$. Assuming $x,y>0$, we may write 
    \begin{equation*}
        \alpha \left(\frac1x +\frac1y\right) >\beta.
    \end{equation*}
    As $x\rightarrow\infty$, we may rewrite this as
    \begin{equation*}
        \frac\alpha\beta > y.
    \end{equation*}
    Since the inequality has to hold for all values $y>0$ supplied by the data, we may plug in $y_{\max}$ and obtain a parameter set that ensures that all estimated curves for the bond price are bounded for all time:
    \begin{equation*}
        \Theta :=\set[\Big]{ (\alpha ,\beta)\given \frac\alpha\beta >y_{\max}}.
    \end{equation*}
    For our application, cross-validation reveals that best results are not within this parameter set. Restricting the parameter search to the set $\Theta$ may thus result in curves suited for extrapolation for long-term maturities with possible slight trade-off in performance on the fitted prices. We do not pursue this treatment in the analysis.
\end{remark}
\subsection{Comparison to standard methodology}
In this section, we want to provide a comparison between our method of kernel regression and a subsequent model reduction and a naive regression using standard methods with a fixed parametric family of exponential curves as the reference model.

For the naive model, we fix the space
\begin{equation*}
	\mathcal E^d:=\left\{ \sum _{i=1}^dc_{i}e^{\lambda _i\cdot}: c_i,\lambda _i\in\R\text{ for }i=1,\dots ,d\right\}.
\end{equation*}
Our goal is to solve the following joint optimisation problem:
\begin{equation*}
	\min _{(g_t)_{t=1}^T\subset\mathcal E^d}\frac{1}{T}\sum _{t=1}^T\vert P_t-C_tg_t\vert ^2
\end{equation*}
for observed price vectors $\{ P_1,\dots ,P_T\}$ and cashflow matrices $\{ C_1,\dots ,C_T\}$. Similarly to the model reduction optimisation scheme, this reduces to a minimisation over the exponents $\{\lambda _1,\dots ,\lambda _d\}$. Unfortunately, this is a non-convex problem with no obvious existence and uniqueness results. To facilitate better fits, we once again add the soft constraint $g(0)=1$ by adding a synthetic contract with tenor $0$ and value $1$. To improve performance, we make the informed choice of starting values close to the optimal exponents observed in the kernel regression. Comparing run times for the naive regression and the kernel regression with subsequent model reduction, we note that while the kernel regression for the full model takes $1-2$ minutes for the full data sample of $1$ year, as well as about $6$ hours for model reduction for all dimensions from $1$ to $30$, the naive regression takes $4-6$ times as long for the $30$-dimensional regression only on the same machine. We compare fitting results in \Cref{fig:naive_comp}. 
\begin{figure}[h!]
	\includegraphics[width=13cm, height=7cm]{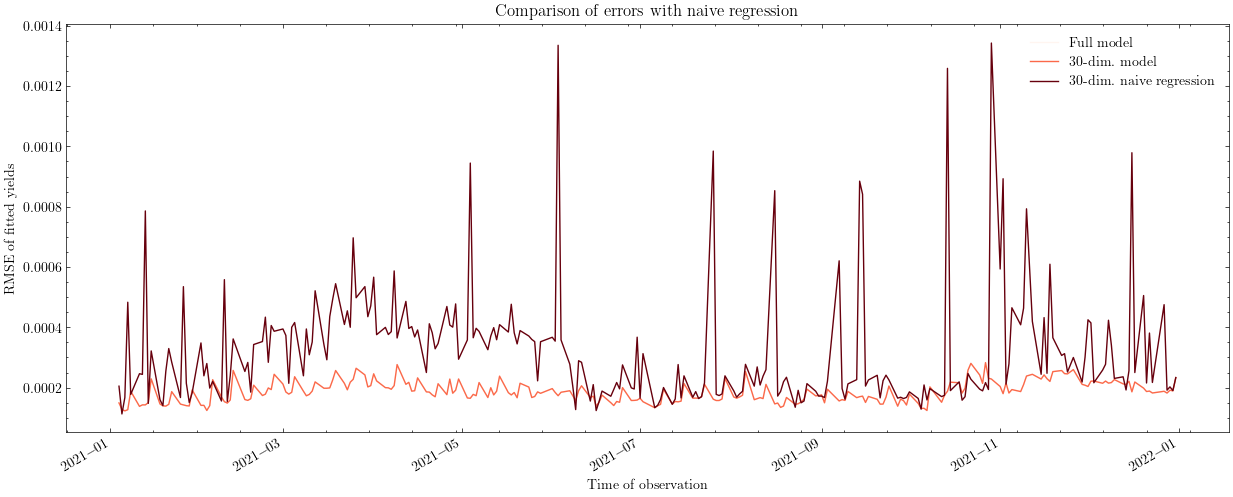}
	\caption{Comparison between kernel regression and naive regression.}\label{fig:simulations4}\label{fig:naive_comp}
\end{figure}
We observe that the naive regression exhibits more erratic behaviour than our approach. In particular, we note that for the observed time range, both the full model and the $30$-dimensional model show smaller root mean square errors for the fitted yields than a naive regression, while also being significantly faster. The $30$-dimensional reduced model is virtually indistinguishable from the full model on the observed time range.
\section{Conclusion and research outlook}\label{sec:conclusion}
In this paper, we have taken the theory of the discount framework due to \cite{filipovic_discount} and used the theory of RKHS to reduce the infinite-dimensional curve fitting problem to a finite-dimensional kernel regression with kernels that are consistent for discount models fulfilling the NAFLVR market condition and a linearity assumption. This allowed us to not only find suitable curve families with very good fit results for the resulting static estimation problem from trading day to trading, but also calibrate an affine stochastic model for the discount. Using a kernel regression and subsequent dimensionality reduction scheme, we were able to fully specify a consistent stochastic model and through the use of a simple bootstrapping method using the extracted covariances of the paths of the diffusion process, as well as the drift implied by our theoretical results, we were able to fully calibrate the model for simulation purposes for arbitrary time frames.

The main focus for this paper was to introduce families of reproducing kernels consistent with the existing discount theory. Various theoretical points still require a more rigorous analysis. In particular, the global existence of the SDE for consistent discount models is not entirely clear due to the presence of quadratic drift. The possibility of more complex models not fulfilling the \ref{LA} specification is also of interest. For the estimation problem and model calibration, investigating more general polynomial-exponential kernels could possibly deliver even better results for both fitting the curves to the available data, as well as in the subsequent calibration of the stochastic model. As observed in our analysis, the proposed dimensionality reduction is not suitable for removing dependence in the extracted paths of the diffusion. Additional steps to further reduce the model to its minimal representation with independent paths would further offer improvements to the specification of the stochastic model and simulations. Finally, adding regularisation terms in the several optimisation steps could potentially enhance the extrapolation properties with respect to the tenor of the curves.

\newpage
\appendix
\section{Reproducing kernel Hilbert spaces}\label{sec:rkhs_theory}
In this section, we provide some of the basic theory of reproducing kernel Hilbert spaces which is used in later parts of the paper. In particular, we provide references for some of the results used throughout the paper. Most of the theory and results contained in this section can be found in \cite{aronszajn} and \cite{paulsen_rkhs}. We begin with the following
\begin{definition}\label{def:4_rkhs}
	Let $\mathcal X$ be a set and $\mathbb F$ denote either $\R$ or $\mathbb C$ and denote by $\mathfrak{F}(\mathcal X,\mathbb F)$ the set of functions $f:\mathcal X\rightarrow\mathbb F$. A subset $\mathcal H\subseteq\mathcal F(\mathcal X,\mathbb F)$ is called a \emph{reproducing kernel Hilbert space} (RKHS) if 
	\begin{enumerate}
		\item $\mathcal H$ is a vector space of $\mathfrak F(\mathcal X,\mathbb F)$.
		\item $\mathcal H$ is a Hilbert space with some inner product $\langle \cdot,\cdot\rangle _{\mathcal H}$.
		\item For every $x\in \mathcal X$, the evaluation functional $\delta _x:\mathcal H\rightarrow\mathbb F$ defined as $\delta _x:f\mapsto f(x)$ is bounded.
	\end{enumerate}
\end{definition}
For the remainder of this section let $\mathcal H\subseteq\mathfrak F(\mathcal X, \mathbb F)$ be an RKHS with inner product $\langle\cdot ,\cdot\rangle _{\mathcal H}$.
In the case of an RKHS, by the Riesz representation theorem and the dual representation of a Hilbert space, there exists a function $k_x\in\mathcal H$, such that
	\begin{equation}\label{eq:4_reproducing}
		f(x)=\langle f,k_x\rangle_{\mathcal H}\quad\text{for any }x\in\mathcal X.
	\end{equation}
\begin{definition}
	The function $k(x,y):=\langle k_x,k_y\rangle _{\mathcal H}$ is called the \emph{reproducing kernel} of $\mathcal H$.
\end{definition}
We also call \Cref{eq:4_reproducing} the \emph{reproducing property} of the reproducing kernel $k$. This property allows us to derive some simple, but powerful results on RKHS, such as denseness of the kernel function (see, e.g.~\cite[Proposition 2.1]{paulsen_rkhs}):
\begin{proposition}\label{prop:4_denseness}
	Let $\mathcal H$ be an RKHS with reproducing kernel $k$ on the set $\mathcal X$. Then the linear span of the set $\{ k_x: x\in \mathcal X\}$ is dense in $\mathcal H$.
\end{proposition}
Once can also derive the following convenient convergence properties, which we will utilise throughout the paper:
\begin{lemma}\label{lem:convergence}
	Let $\mathcal H$ be an RKHS, let $f\in\mathcal H$ and let $(f_n)_{n\geq 0}\subseteq\mathcal H$ be a sequence of functions in $\mathcal H$.
	\begin{enumerate}
		\item If $f_n\rightarrow f$ in norm, then $f_n\rightarrow f$ pointwise.
		\item If $f_n\rightarrow f$ weakly, then $f_n\rightarrow f$ pointwise. If, additionally $\Vert f_n\Vert _{\mathcal H}< \infty$ for all $n\geq 0$, then $f_n\rightarrow f$ pointwise implies $f_n\rightarrow f$ weakly. 
	\end{enumerate}
\end{lemma}
\begin{proof}
    See \Cref{app:b}.
\end{proof}

We provide one of the more important characterisations of functions which are candidates to become reproducing kernels of some RKHS.
\begin{definition}\label{def:4_kernel_function}
	Let $\mathcal X$ be a set and let $f:\mathcal X\times \mathcal X\rightarrow\mathbb F$ be a symmetric function, that is $f(x,y)=f(y,x)$ for all $x,y\in\mathcal X$. $f$ is called \emph{positive semidefinite} if for any $n\in\mathbb N$ and any choice $\{ x_1,\dots ,x_n\}\subseteq \mathcal X$ of distinct $n$ points the matrix $(k(x_i,x_j))_{i,j=1}^n$ is positive semidefinite. In that case, we will also refer to $f$ as a \emph{kernel function}.
\end{definition}
\begin{remark}
	Note that positive semidefiniteness implies that $f(x,y)\geq 0$ for all $x,y\in \mathcal X$ and if $\mathbb F =\mathbb C$, then it also implies that $f$ is symmetric.
\end{remark}
Kernel functions are fundamental in the sense that they correspond to reproducing kernel of some RKHS. Indeed, we have the following two important results:
\begin{proposition}
	Let $\mathcal H$ be an RKHS on the set $\mathcal X$ with reproducing kernel $k$. Then $k$ is a kernel function.
\end{proposition}
\begin{proof}
	We have for any $x,y\in \mathcal X$ that $k(x,y)=\langle k_x,k_y\rangle _{\mathcal H}$. Therefore, symmetry and positive definiteness follows from the properties of the inner product.
\end{proof}
The converse statement is a fundamental result due to Moore (see, e.g. \cite{aronszajn}, \cite[Theorem 2.14.]{paulsen_rkhs}).
\begin{theorem}\label{thm:4_aronszajn_moore}
	Let $\mathcal X$ be a set and let $k:\mathcal X\times \mathcal X\rightarrow\mathbb F$ be a kernel function. Then there exists a unique RKHS $\mathcal H$ such that $k$ is the reproducing kernel for $\mathcal H$.
\end{theorem}

\section{Technical tools}\label{app:b}
\begin{lemma}\label{lem:dx_invariance}
	A finite-dimensional vector space $V\subseteq C^1(\mathbb R_+,\mathbb R)$ is derivative invariant if and only if there is a matrix $A$ such that the coordinate functions of $b(x):=\exp (xA)b_0$ are a basis for $V$ where $b_0:=b(0)$.
\end{lemma}
\begin{proof}
	Let $V\subseteq C^1(\R _+,\R)$ be a finite-dimensional vector space. Assume there are $A\in\R ^{d\times d}$ for some $d\in\mathbb N$ and $b:\R _+\rightarrow\R ^d$ defined as above such that $\{ b(x): x\in\R _+\}$ are a basis for $V$, then we see that for each $k=1,\dots,d$, $b_k' = (Ab)_k = \sum _{l=1}^d A_{kl}b_l \in V$. Thus, the derivatives of the basis functions $b_1,\dots,b_d$ are in $V$ and, hence, $V$ is derivative invariant.
For the converse, assume that $V$ is derivative invariant. Let $b_1,\dots,b_d$ be a basis for $V$. Then $b_k' \in V$ and hence there are coefficients $a_{k,1},\dots,a_{k,d}$ such that $b_k' = \sum_{l=1}^d a_{k,l}$. Define $A:=(a_{k,l})_{kl}$. Observe that $b' = Ab$. We find that $b(x) = \exp (xA) b_0$ where $b_0=b(0)$.
\end{proof}
\begin{corollary}\label{cor:full_consistency_equiv}
	Let $k$ be a kernel function in the sense of \Cref{def:4_kernel_function}. $k$ is fully consistent if and only if for any $y_1,\dots, y_N$ there is a matrix $A_N$ such that $b(x) := (k(y_1,x),\dots,k(y_N,x))$ satisfies $b(x) = \exp( xA )b_0$.
\end{corollary}
\begin{proof}
	This follows immediately from \Cref{def:fully consistent_kernel} and \Cref{lem:dx_invariance} by taking $V=\text{span}\{k_{y_1},\dots ,k_{y_N}\}$.
\end{proof}
\begin{proposition}\label{prop:4_bargmann_space}
	Let $k_{\exp}$ be the exponential kernel on $\mathbb{C}$, that is $k_{\exp}=e^{x\bar y}$ and let $\mathcal H(\exp )$ denote the RKHS induced by $k_{\exp}$. Denote by $\mathcal B$ the Segal-Bargmann space of analytic functions defined as
	\begin{equation*}
		\mathcal B:=\set{ f:\mathbb C\rightarrow \mathbb C\given \Vert f\Vert _{\mathcal B} <\infty },
	\end{equation*}
	where the norm $\Vert\cdot\Vert _{\mathcal B}$ is induced by the inner product
	\begin{equation*}
		\langle f,g\rangle _{\mathcal B}:=\frac{1}{\pi}\int _{\mathbb C}f(z)\overline{g(z)}e^{-\vert z\vert ^2}\frac{i}{2}dz\wedge d\bar z,\qquad f,g\in\mathcal B.
	\end{equation*}
	Then $\mathcal B =\mathcal H(\exp )$.
\end{proposition}
\begin{proof}
	It is sufficient to prove that $\langle f,g\rangle _{\mathcal H(\exp )}=\langle f,g\rangle _{\mathcal B}$ for all $f,g\in\mathcal B$. We first note that the surface $1$-form $dz\wedge d\bar z$ satisfies 
	\begin{equation*}
	    \begin{aligned}	
                dz\wedge d\bar z&=(dx+idy)\wedge (dx-idy)\\
                &=dx\wedge dx+idy\wedge dx-idx\wedge dy-dy\wedge dy=-2idx\wedge dy
            \end{aligned}
	\end{equation*}
	by the skew-symmetry of the wedge product, and therefore we have
	\begin{equation*}
		\begin{aligned}
			\langle f,g\rangle _{\mathcal B}&=\frac{1}{\pi}\int _{\mathbb C}f(z)\overline{g(z)}e^{-\vert z\vert ^2}\frac{i}{2}dz\wedge d\bar z\\
            &=\frac{1}{\pi}\int _0^{\infty}\int _0^{\infty}f(x+iy)\overline{g(x+iy)}e^{-(x^2+y^2)}dxdy.
		\end{aligned}
	\end{equation*}
	Making the coordinate change $x+iy=re^{it}$, we find
	\begin{equation*}
		\begin{aligned}
			\langle f,g\rangle _{\mathcal B}&=\frac{1}{\pi}\int _0^{\infty}\int _0^{2\pi}f(re^{it})\overline{g(re^{it})}re^{-r^2}dtdr\\
							&=\frac{1}{\pi}\int _0^{\infty}\int _0^{2\pi}\sum _{k=0}^{\infty}\sum _{l=0}^{\infty}\frac{f^{(k)}(0)}{k!}\frac{\overline{g^{(l)}(0)}}{l!}\left( re^{it}\right) ^k\left( re^{-it}\right) ^lre^{-r^2}dtdr\\
							&=\frac{1}{\pi}\sum _{k=0}^{\infty}\sum _{l=0}^{\infty}\frac{f^{(k)}(0)}{k!}\frac{\overline{g^{(l)}(0)}}{l!}\int _0^{\infty}r^{2k+1}e^{-r^2}\int _0^{2\pi}e^{it(k-l)}dtdr,
		\end{aligned}
	\end{equation*}
	where we used the analyticity of the functions $f$ and $g$ to write down their Taylor series expansions and interchange summation with integration. Now we note that
	\begin{equation*}
		\begin{aligned}
			\int _0^{2\pi}e^{it(k-l)}dt =\begin{cases}2\pi, &\text{ if }k=l,\\
			0, &\text{ else.}\end{cases}
		\end{aligned}
	\end{equation*}
	We additionally have that
	\begin{equation*}
		\begin{aligned}
			\int _0^{\infty}r^{2k+1}e^{-r^2}dr=\frac{\Gamma (k+1)}{2} = \frac{k!}{2}
		\end{aligned}
	\end{equation*}
	for $k\in\mathbb N$. Therefore, 
	\begin{equation*}
		\langle f,g\rangle _{\mathcal B}=\sum _{k=0}^{\infty}\frac{f^{(k)}(0)\overline{g^{(k)}(0)}}{k!},
	\end{equation*}
	which proves the assertion.

\end{proof}
\begin{lemma}\label{lem:dx_product}
	Let $k_{\exp}$ be the exponential kernel on $\mathbb C$ and $\mathcal H(\exp )$ be the induced RKHS. Then we have
	\begin{equation*}
		\langle\partial _{\lambda } (\cdot ^ke^{\lambda\cdot}),\partial _{\mu}(\cdot ^le^{\mu\cdot})\rangle _{\mathcal H(\exp )}=\partial _{\lambda}\partial _{\mu}\langle e^{\lambda\cdot},e^{\mu\cdot}\rangle _{\mathcal H(\exp )}
	\end{equation*}
\end{lemma}
\begin{proof}
	To simplify the proof, we will use \Cref{prop:4_bargmann_space} and work with the integral representation of the inner product. We have
	\begin{equation*}
		\begin{aligned}
			&\int _0^{\infty}\int _0^{\infty}\partial _{\lambda}\partial _{\mu}\left( (x+iy)^ke^{\lambda (x+iy)}\overline{(x+iy)^le^{\mu (x-iy)}}e^{-(x^2+y^2)}\right) dxdy\\
			&=\int _0^{\infty}\int _0^{\infty}(x+iy)^k(x-iy)^l\partial _{\lambda}\left( e^{\lambda (x+iy)}\right)\partial _{\mu}\left(\overline{e^{\mu (x+iy)}}\right)e^{-(x^2+y^2)}dxdy\\
			&=\int _0^{\infty}\int _0^{\infty} (x+iy)^{k+1}(x-iy)^{l+1}e^{\lambda (x+iy)}e^{\mu (x-iy)}e^{-(x^2+y^2)}dxdy\\
			&=\int _0^{\infty}\int _0^{2\pi}r^{k+1}e^{itk}r^{l+1}e^{-itl}e^{\lambda re^{itk}}e^{\mu re^{-itl}}e^{-r^2}dtdr\\
			&=\int _0^{\infty}\int _0^{2\pi}r^{k+l+2}e^{it(k-l)}e^{r(\lambda e^{itk}+\mu e^{-itl})}e^{-r^2}dtdr\leq 2\pi\int _0^{\infty}r^{k+l+2}e^{r(\mu +\lambda)-r^2}dr.
		\end{aligned}
	\end{equation*}
	For any $\lambda ,\mu\in\R$ we can find a constant $C>0$ such that $e^{r(\lambda +\mu)}<Ce^{r^2/2}$ and thus
	\begin{equation*}
		\begin{aligned}
			&\int _0^{\infty}\int _0^{\infty}\partial _{\lambda}\partial _{\mu}\left( (x+iy)^ke^{\lambda (x+iy)}\overline{(x+iy)^le^{\mu (x-iy)}}e^{-(x^2+y^2)}\right) dxdy\\
			&\leq 2\pi C\int _0^{\infty}r^{k+l+2}e^{-r^2/2}dr<\infty.
		\end{aligned}
	\end{equation*}
	Since we have found an integrable majorant, we may use the Leibniz integral rule (see, e.g. \cite{folland}) to interchange differentiation and integration, yielding the result.

\end{proof}
Before we provide the next result, we first provide a definition of the notiotion of \emph{coercive} functions for the reader's convenience.
\begin{definition}\label{def:coercive}
	Let $E$ be a Banach space. A function $f:E\rightarrow\R$ is called \emph{coercive} if for any sequence $\{x _n\}_{n\in\mathbb N}\subseteq E$ with $\Vert x_n\Vert\rightarrow\infty$ as $n\rightarrow\infty$, it holds that $f(x_n)\rightarrow\infty$ as $n\rightarrow\infty$.
\end{definition}
\begin{lemma}\label{lem:weak_optimisation}
	Let $E$ be a reflexive Banach space and let $q:E\rightarrow\R$ be a continuous, coercive, strictly convex map. Assume $V\subseteq E$ is a non-empty set which is closed in the weak topology $\sigma (E, E^*)$. Then there is $x_0\in V$, such that 
	\begin{equation*}
		q(x_0)=\inf\lbrace q(x): x\in V\rbrace.
	\end{equation*}
\end{lemma}
\begin{proof}
	Since $q$ is coercive, the set $C_a:=\{ x\in E: q(x)\leq a\}$ for $a\in\R$ is bounded. Furthermore, since $C_a$ is closed and convex, there is $a\geq 0$, such that $C_a\cap V\neq\emptyset$. By the Banach-Alaoglu Theorem (cf. \cite[Theorem 3.16]{brezis}), $C_a$ is weakly compact, hence, $C_a\cap V$ is weakly compact by the closedness of $V$. Since $q$ is convex and continuous in the norm topology on $E$, it is lower semicontinuous in the weak topology. Thus, there is $x_0\in C_a\cap V$, such that $q(x_0)=\inf _{x\in C_a\cap V}q(x)=\inf _{x\in V}q(x)$ as required.
\end{proof}

\begin{lemma}\label{lem:q_factorisation}
	Let $p:\R _+\times\R _+\rightarrow\R$ be a symmetric, positive semidefinite polynomial. Assume there is $x_0\geq 0$ such that $p(x_0,x_0)=0$. Then there is a positive semidefinite polynomial $r:\R _+\times\R _+\rightarrow\R$ such that
	\begin{equation}\label{eq:q_factorisation1}
		p(x,y)=(x-x_0)(y-x_0)r(x,y)\qquad\text{for all }x,y\geq 0.
	\end{equation}
\end{lemma}
\begin{proof}
	Define the rational positive semidefinite function
	\begin{equation*}
		r(x,y):=\frac{p(x,y)}{(x-x_0)(y-x_0)}\qquad\text{for all }x,y\in\R _+\setminus\{ x_0\}.
	\end{equation*}
	Let $y\in\R _+\setminus\{ x_0\}$ be fixed. To conclude the proof, we will show that $r$ is a polynomial. Note that $\vert p(x_0,y)\vert <\sqrt{p(x_0,x_0)}\sqrt{p(y,y)}=0$. Thus, $p(\cdot, y)$ is a polynomial with a root in $x_0$, and therefore $r(\cdot ,y)$ is a polynomial. Now, $r$ is a polynomial function in its first argument and, by symmetry, its second argument. This implies it is a polynomial in two variables (see, e.g. \cite{carroll}).
\end{proof}
\begin{lemma}\label{lem:q_factorisation2}
	Let $p:\R _+\times\R _+\rightarrow\R$ be a positive semidefinite polynomial. Then there is a polynomial $q:\R_+\rightarrow\R$ and a positive semidefinite polynomial $r:\R _+\times\R _+\rightarrow\R$ such that
	\begin{enumerate}
		\item $p(x,y)=q(x)q(y)r(x,y)$ for all $x,y\geq 0$,
		\item $r(x,x)>0$ for all $x\geq 0$.
	\end{enumerate}
\end{lemma}
\begin{proof}
	This is a direct consequence of \Cref{lem:q_factorisation} through repeated factorisation.
\end{proof}
\begin{lemma}\label{lem:q_multiplier}
	Let $p:\R _+\times\R _+\rightarrow\R$ be a symmetric, positive semidefinite polynomial. Let $q:\R_+\rightarrow\R$ and $r:\R _+\times\R _+\rightarrow\R$ be as in \Cref{lem:q_factorisation2} and let $\mathcal H(r)$ and $\mathcal H(p)$ be the RKHS induced by $r$ and $p$, respectively. For a function $f:\R_+\rightarrow\mathbb C$ define $M_q(f):=q\cdot f$. Then $M_q:\mathcal H(r)\rightarrow\mathcal H(p)$ is a linear bijective isometry.
\end{lemma}
\begin{proof}
	Let $h\in\mathcal H(r)$. Since $r$ is a polynomial, we have $\text{dim}(\mathcal H(r))=\text{deg}(r)<\infty$. Then, for some $n<\infty$, there are $\lambda _1,\dots ,\lambda _n$, $x_1,\dots ,x_n\geq 0$ with $q(x_i)\neq 0$ for $i=1,\dots ,n$ such that $h=\sum _{i=1}^n\lambda _ir(\cdot ,x_i)$. We have
	\begin{equation*}
		M_q(h)=q\cdot h=\sum _{i=1}^n\lambda _iq\cdot h(\cdot, x_i)=\sum _{i=1}^n\frac{\lambda _i}{q(x_i)}p(\cdot, x_i)\in\mathcal{H}(p).
	\end{equation*}
	Thus, $M_q:\mathcal H(r)\rightarrow\mathcal H(p)$ is linear. Let now $g\in\mathcal H(p)$. Then, for $m<\infty$ there are $\eta _1,\dots ,\eta _m\in\R$ and $x_1,\dots ,x_m\geq 0$ such that $g=\sum _{i=1}^m\eta _ip(\cdot ,x_i)$. We then have
	\begin{equation*}
		\begin{aligned}
			g&=\sum _{i=1}^m\eta _ip(\cdot ,x_i)=\sum _{i=1}^m\eta _iq\cdot q(x_i)r(\cdot ,x_i)\\
			 &=q\cdot\sum _{i=1}^m\left(\eta _iq(x_i)\right) r(\cdot ,x_i)=M_q\underbrace{\left(\sum _{i=1}^m\left(\eta _iq(x_i)\right) r(\cdot ,x_i)\right)}_{\in\mathcal H(r)}.
		\end{aligned}
	\end{equation*}
	Hence, $M_q$ is surjective. It is easy to see that $M_q(f)=0$ if and only if either $f\equiv 0$ or $q\equiv 0$, hence it is injective and therefore bijective.

	For $h\in\mathcal H(r)$ we have 
	\begin{equation*}
		M_q(h)(x)=q(x)h(x)=q(x)\langle h,r(\cdot ,x)\rangle _{\mathcal H(r)}=\langle h,q(x)r(\cdot ,x)\rangle _{\mathcal H(r)}.
	\end{equation*}
	On the other hand,
	\begin{equation*}
		\begin{aligned}
			M_q(h)(x)&=q(x)h(x)=(q\cdot h)(x)=\langle q\cdot h,p(\cdot ,x)\rangle _{\mathcal H(p)}\\
				 &=\langle M_q(h), p(\cdot ,x)\rangle _{\mathcal H(p)}=\langle h,M_q^*(p(\cdot ,x))\rangle _{\mathcal H(r)}.
		\end{aligned}
	\end{equation*}
	Therefore, we have $M_q^*(p(\cdot ,x))=r(\cdot ,x)q(x)=p( \cdot, x)/q$ and hence $M_q^*(g)=g/q$ for all $g\in\mathcal H(p)$. This implies $M_q^*=M_q^{-1}$, whence $M_q$ is an isometry.
\end{proof}
\begin{proposition}\label{prop:4_isometry}
	Let $p:\R _+\times\R _+\rightarrow\R$ be a symmetric, positive semidefinite polynomial. Let $q$ and $r$ be defined as in \Cref{lem:q_factorisation2} and $\mathcal H(p)$ and $\mathcal H(r)$ be the RKHS induced by $p$ and $r$, respectively. Furthermore, let $\mathcal H(\exp)$ be the RKHS induced by the exponential kernel. Then
	\begin{equation*}
		H(p)\otimes\mathcal H(\exp)\cong\mathcal H(r)\otimes\mathcal H(\exp)
	\end{equation*}
	isometrically with isometry 
	\begin{equation*}
		\begin{aligned}
			&M_q:\mathcal H(r)\otimes\mathcal H(\exp )\rightarrow\mathcal H(p)\otimes\mathcal H(\exp)&\\
			&M_q:a\otimes b\mapsto (q\cdot a)\otimes b&\text{for all }a\in\mathcal H(r), b\in\mathcal H(\exp).
		\end{aligned}
	\end{equation*}
\end{proposition}
\begin{proof}
	This is a direct consequence of \Cref{lem:q_multiplier}.
\end{proof}
\begin{corollary}\label{cor:isometry}
	Let $p,q,r$ and $\mathcal H(p),\mathcal H(r),\mathcal H(\exp)$, as well as $M_q$ be defined as in \Cref{prop:4_isometry}. Let $\emph{tr}$ be the map $\Tr: f\otimes g\mapsto f\cdot g$. Then the following diagram commutes
	\begin{equation*}
		\begin{tikzcd}
			\mathcal H(r)\otimes \mathcal H(\exp) \arrow{r}{M_q} \arrow[swap]{d}{\Tr} & \mathcal H(p)\otimes\mathcal H(\exp) \arrow{d}{\Tr} \\
			\mathcal H(r)\mathcal H(\exp) \arrow{r}{M_q}& \mathcal H(p)\mathcal H(\exp)
		\end{tikzcd}
	\end{equation*}
	.
\end{corollary}
\begin{proof}
	This follows directly from \Cref{prop:4_isometry}.
\end{proof}

\section{Technical proofs}\label{app:c}

\begin{proof}[Proof of \Cref{prop:4_1}]
        The dynamics of the discounted price process $\tilde{P}(t,T)$ as defined in \Cref{eq:4_3}, fulfill
	\begin{equation}\label{eq:a1}
                d\tilde{P}(t,T)=-r_te^{-\int _0^tr_sds}P(t,T)dt+e^{-\int _0^tr_sds}dP(t,T).
        \end{equation}
	Since $H_t$ is a mild solution to \Cref{eq:4_sde}, we may use \Cref{eq:4_2} and write
	\begin{equation*}
		H_t(T-t) = H_0(T)+\int _0^t\alpha _s(T-s)ds+\int _0^t\Sigma _s(T-s)dW_s.
	\end{equation*}
	Using the fact that $P(t,T)=1-H_t(T-s)$ and inserting the dynamics into \Cref{eq:a1} yields
	\begin{equation*}
		d\tilde{P}(t,T)=e^{-\int _0^tr_sds}\left( (-\alpha _t(T-t)-r_t(1-H_t(T-t))dt-\Sigma _t(T-t)dW_t\right).
	\end{equation*}
        Now, the process $\tilde{P}(t,T)$ is a local martingale if and only if its drift vanishes. This is equivalent to
	\begin{equation*}
		\alpha _t=r_tH_t-r_t.
	\end{equation*}
	If, in addition, $H_t$ is a strong solution to \eqref{eq:4_sde} and the evaluation functional $\delta _xh:=h(x)$ is continuous, the process $\beta _t(x):=\delta _x\left(\partial _xH_t+\alpha _t\right)$ exists for all $x\in\R_+$ and is the drift of the process $H_t(x)$. Then we obtain the drift condition
	\begin{equation*}
		\beta _t(x)=H_t'(x)-r_t+r_tH_t(x).
	\end{equation*}

\end{proof}
\begin{proof}[Proof of \Cref{prop:4_3}]
	In the case of a $d$-dimensional diffusion $Y$, we find that $\beta _t(x)=\langle g(x),D_yf(Y_t)b_t+\frac12\sum _{k=1}^d\Tr\left( \sigma _t\sigma _t^{\top}H_yf_k(Y_t)\right) e_k\rangle$. The drift condition \eqref{eq:4_5} now reads
        \begin{equation*}
		\langle g(x), D_yf(Y_t)b_t+\frac12\sum _{k=1}^d\Tr\left( \sigma _t\sigma _t^{\top}H_yf_k(Y_t)\right) e_k\rangle = H_t'(x)-r_t+r_tH_t(x).
        \end{equation*}
        Inserting our definitions yields
        \begin{equation*}
                \begin{aligned}
			&\langle g(x),D_yf(Y_t)b_t+\frac12\sum _{k=1}^d\Tr\left( \sigma _t\sigma _t^{\top}H_yf_k(Y_t)\right) e_k\rangle \\ 
			&=\langle g'(x),f(Y_t)\rangle -\langle g'(0),f(Y_t)\rangle +\langle g'(0),f(Y_t)\rangle\langle g(x),f(Y_t)\rangle\\
		\end{aligned}
        \end{equation*}
Grouping terms in a convenient way yields
\begin{equation}\label{eq:4_10}
	\begin{aligned}
		&\langle g(x), D_yf(Y_t)b_t+\frac12\sum _{k=1}^d\Tr\left( \sigma _t\sigma _t^{\top}H_yf_k(Y_t)\right) e_k-\langle g'(0), f(Y_t)\rangle f(Y_t)\rangle \\
		&=\langle g'(x)-g'(0),f(Y_t)\rangle .
	\end{aligned}
\end{equation}
We now observe that since we have an inner product of a function depending only on $x$ and a random process independent of $x$, we may look at the factors separately. That is, there is a vector $a_0\in\R ^d$ and a matrix $A\in\R ^{d\times d}$, such that the following system of equations holds
\begin{equation*}
	\begin{aligned}
		\langle a_0,\tilde{g}(x)\rangle &=g'_0(x)+g'_0(0)g_0(x)-g'_0(0),\\
		A\tilde{g}(x)&=\tilde{g}'(x)+g'(0)g_0(x)+g'_0(0)\tilde{g}(x)-g'(0),
	\end{aligned}
\end{equation*}
where $\tilde{g}:=(g_1,\dots ,g_d)^{\top}$. After reordering terms, We obtain the inhomogeneous linear system of ODEs
\begin{equation}\label{eq:ode_nonmat}
	\begin{aligned}
		g'_0(x)&=\langle a_0,\tilde{g}(x)\rangle -g'_0(0)g_0(x)+g'_0(0),\\
		\tilde{g}'(x)&=-\tilde{g}'(0)g_0(x)+A\tilde{g}(x)+-g'_0(0)\tilde{g}(x)+\tilde{g}'(0).
	\end{aligned}
\end{equation}
Define the matrix
\begin{equation*}
	M:=\begin{pmatrix} -g'_0(0) & a_0^{\top} \\
	-\tilde{g}'(0) & A-g'_0(0)\mathbbm{1}_d \end{pmatrix}.
\end{equation*}
We can therefore write \Cref{eq:ode_nonmat} as
\begin{equation}\label{eq:ode_mat}
	g'(x)=Mg(x)+g'(0)
\end{equation}
If $M$ is non-singular, using the constraint $g(0)=0$, the solution can be written as
\begin{equation*}
	g(x)=M^{-1}(e^{xM}-\mathbbm{1}_{d+1})g'(0).
\end{equation*}
Otherwise, we may formally define the function $\xi(z):=\frac{e^{z}-1}{z}$ and observe that its power series satisfies
\begin{equation*}
        \xi(\lambda z)=\frac{e^{\lambda z}-1}{z}=\sum _{k=0}^{\infty}\frac{\lambda ^kz^k}{(k+1)!}.
\end{equation*}
We may therefore write the solution as
\begin{equation}\label{eq:xi_fun}
        g(x)=x\xi(xM)g'(0).
\end{equation}
Note now, that $g'(0)=-Me_0$, where $e_0\in\R^{d+1}$ denotes the first unit basis vector with the extended vector notation. Therefore 
\begin{equation*}
	g(x)=-x\xi(xM)Me_0=-(e^{xM}-\mathbbm{1}_{d+1})e_0=(\mathbbm{1}_{d+1}-e^{xM})e_0
\end{equation*}
as asserted.
Using the separability argument on \Cref{eq:4_10}, we find, after using the fact that $\langle Mx,y\rangle = \langle x,M^{\top}y\rangle$,
\begin{equation*}
    \begin{aligned}
	D_yf(Y_t)(0,b_t)^{\top}+\frac12\sum _{k=1}^d\Tr\left( \sigma _t\sigma _t^{\top}H_yf_k(Y_t)\right) e_k&\\
    -\langle g'(0),f(Y_t)\rangle f(Y_t)&= M^{\top}f(Y_t)
    \end{aligned}
\end{equation*}
Grouping up terms yields the asserted drift condition. For the converse, assume we are given a process $Z:=(1,f(Y))$ fulfilling the drift condition for some diffusion process $Y$ with dynamics $dY_t=b_tdt+\sigma _tdW_t$ and function $f:\R\rightarrow\R$ and a function $g:\R _+\rightarrow\R ^{d+1}$ satisfying $g(x)=(\mathbbm 1_{d+1}-e^{xM})e_0$ for some matrix $M\in\mathbb R^{(d+1)\times (d+1)}$. It is an easy calculation to verify that the process $H:=\langle g(x),Z_t\rangle$ satisfies the drift condition \eqref{eq:4_5} and thus fulfills the equivalent conditions of \Cref{prop:4_3}. 
\end{proof}

\begin{proof}[Proof of \Cref{prop:4_time_inhom_case}]
	Let $H_t(x)=\langle g(x,t),f(Y_t)\rangle$. Proceeding the same way as in the proof of \Cref{prop:4_3} and using \Cref{eq:4_drift_time_inhom}, we obtain
	\begin{equation*}
		\begin{aligned}
			&\langle g(x,t),D_yf(Y_t)b_t+\frac12\sum _{k=1}^d\Tr\left(\sigma _t\sigma _t^{\top}H_yf(Y_t)\right) e_k-\langle \partial _xg(0,t),f(Y_t)\rangle f(Y_t)\rangle\\
			&= \langle\partial _xg(x,t)-\partial _tg(x,t)-\partial _xg(0,t),f(Y_t)\rangle.
		\end{aligned}
	\end{equation*}
	Using the separability argument, we immediately obtain that \Cref{eq:4_prop_time_inhom2} holds and that the function $g$ satisfies
	\begin{equation}\label{eq:4_time_inhom_pde}
		\partial _xg(x,t)-\partial _tg(x,t)=M(t)g(x,t)+\partial _xg(0,t),
	\end{equation}
	where $M$ is defined as
	\begin{equation*}
		M(t):=\begin{pmatrix} -\partial _xg_0(0,t) & a_0^{\top} \\
		-\partial_x\tilde{g}(0,t) & A-\partial_xg_0(0,t)\mathbbm{1}_d \end{pmatrix}
	\end{equation*}
	for some matrix $A\in\R ^{d\times d}$ and vector $a_0\in\R^d$, and where $\tilde{g}:=(g_1,\dots ,g_d)^{\top}$. Using the terminal condition $H_t(0)=0$, we furthermore note that $g(0,t)=0$ for all $t$, thus proving the assertion.
\end{proof}
\begin{proof}[Proof of \Cref{cor:4_time_inhom_expl}]
	Assume $g$ satisfies \Cref{eq:4_time_inhom_pde}. Differentiating in the variable $x$, and using the fact that $g$ has continuous second derivatives so that $\partial _{xt}g(x,t)=\partial _{tx}g(x,t)$, we may define $u(x,t):=\partial _xg(x,t)$ and obtain the equation
	\begin{equation*}
		\partial _xu(x,t)-\partial _tu(x,t)=M(t)u(x,t).
	\end{equation*}
	Using standard PDE solution theory (e.g. the method of characteristics), and that the matrix $M(t)$ commutes with its integral, we obtain that 
	\begin{equation*}
		u(x,t)=e^{\int _0^xM(x+t-\xi)d\xi}c(x+t)
	\end{equation*}
	for some unspecified function $c:\R_+\rightarrow\R^d$. Since $g$ is the antiderivative of $u$, we may integrate with respect to the variable $x$ and use the terminal condition $g(0,t)=0$, to obtain the solution specified in \Cref{eq:4_pde_expl_sol}. It is now easy to check that $g$ of this form satisfies \Cref{eq:4_time_inhom_pde}.
\end{proof}
\begin{proof}[Proof of \Cref{prop:4_u_space}]
	This is a direct consequence of \Cref{cor:full_consistency_equiv}.
\end{proof}
\begin{proof}[Proof of \Cref{lem:v_space_u}]
	Let $g\in\mathcal{U}$ and $\lambda\in\R$. Recall from \Cref{eq:xi_fun} and \Cref{prop:4_u_space} that $g(x)=\varphi (x\xi (xM)v)$ for some $\varphi\in (\R ^{d+1})^*, M\in\R ^{(d+1)\times (d+1)}, v\in\R ^{d+1}$. We have by linearity of the dual space
		\begin{equation*}
			\lambda g(x)=\lambda\varphi (x\xi (xM)v)=\varphi (x\xi (xM)\lambda v)=\varphi (x\xi (xM)w,
		\end{equation*}
		where $w=\lambda v$. Therefore, $\lambda g\in\mathcal{U}$.
		To prove that $\mathcal{U}$ is closed under summation we consider $f,g\in\mathcal{U}$, where $g=\varphi (x\xi (xM_1)v)$ for some $\varphi\in (\R ^{d_1+1})^*, M_1\in\R ^{(d_1+1)\times (d_1+1)}, v\in\R ^{d_1+1}$ and $f(x)=\psi (x\xi (xM_2)w)$ for some $\psi\in (\R ^{d_2+1})^*, M_2\in\R ^{(d_2+1)\times (d_2+1)}, v\in\R ^{d_2+1}$, as well as the direct sum $\R ^{d_1+1}\oplus\R ^{d_2+1}$ with the canonical isomorphism $s:(v,w)\mapsto v+ w$. By the properties of block-diagonal matrices, we observe that $\xi (x(M_1\oplus M_2))=(\xi (xM_1))\oplus (\xi (xM_2))$ and from this we see that
		\begin{equation*}
                \begin{aligned}
			f(x)+g(x)&=\varphi (x\xi(xM_1)v)+\psi (x\xi(xM_2) w)\\
                        &=s^*\left((\varphi\oplus\psi)(x\xi (x(M_1\oplus M_2))(v\oplus w))\right).
		      \end{aligned}
            \end{equation*}
		Therefore, by the canonical isomorphism, $h:=f+g\in\mathcal{U}$ and $h(x)=\vartheta (x\xi (xM_3)u)$, where $(\R ^{d_1+d_2})^*\ni\vartheta :=s^*\circ(\varphi\oplus\psi)$, $\R ^{(d_1+d_2+2)\times (d_1+d_2+2)}\ni M_3:=M_1\oplus M_2$, and $\R ^{d_1+d_2+2}\ni u:=v\oplus w$. Finally, since the function $g(x):=x\xi (z)v$ is smooth, each element $f\in\mathcal{U}$ is of class $C^{\infty}$.
\end{proof}

\begin{proof}[Proof of \Cref{lem:rkhs_general}]
	We will first show that $k$ as defined in the Lemma is a kernel function. We will rely on some of the basic results presented in \cite{aronszajn} and \cite{paulsen_rkhs} without further explicitly referencing them. 
	
	We note that $k$ is of the form $k=h\circ\tilde k\circ(\varphi\otimes\varphi)$, where $\varphi: x\mapsto ax-c$ and $\tilde k(x,y):=xy$. Now, since $h$ a real analytic function, its power series expansion has infinite radius of convergence, is unique and is given by its Taylor series $h(x)=\sum _{k=0}^{\infty}h^{(k)}(0)/k!x^k$. By assumption, $h^{(k)}(0)\geq 0$ and therefore, by \cite[Theorem 4.16]{paulsen_rkhs}, the function $\tilde k$ is a kernel defined on $\R$. Thus, by \cite[Proposition 5.6]{paulsen_rkhs}, $k$ is a kernel.

	We obtain a description of the RKHS $\mathcal{H}(k)$ induced by the kernel $k$ by considering again $k=h\circ\varphi$. The RKHS induced by $k_h(x,y):=h(xy)$ may be obtained from the results in \cite[Theorem 7.2]{paulsen_rkhs}. Indeed, a function $f$ belongs to $\mathcal{H}(k_h)$ if and only if it is of the form 
	\begin{equation}\label{eq:segal_bergman_gen}
		f(x)=\left\langle\bigoplus _{k=0}^{\infty}\sqrt{\frac{h^{(k)}(0)}{k!}}x^k, w\right\rangle _{\mathcal F(\R)},
	\end{equation}
	for some $w\in \mathcal F(\R)$. Here, $\mathcal F(\mathcal L)=\R\oplus\mathcal L \oplus\mathcal L^{\otimes 2}\oplus\dots$ denotes the Fock space over a Hilbert space $\mathcal L$ (see, e.g. \cite[Definition 7.1]{paulsen_rkhs}). In our case, since $\mathcal L =\R$, we have $\mathcal{F}(\R )\cong\ell ^2(\R )$. Thus, the description in \Cref{eq:segal_bergman_gen} corresponds to any function which can be written as a convergent power series with infinite convergence radius $f(x)=\sum _{k=0}^{\infty}a_kx^k$ with coefficients $(a_k)_{k\geq 0}$ fulfilling the following conditions
	\begin{enumerate}
		\item $a_k=0$ for all $k\geq 0$ such that $h^{(k)}(0)=0$.
		\item The weighted square sum is finite, i.e.
			\begin{equation*}
				\sum _{k=0}^{\infty}\mathbbm{1}_{\lbrace h^{(k)}(0)>0\rbrace}\frac{k!}{h^{(k)}(0)}\vert a_k\vert ^2 <\infty. 
			\end{equation*}
	\end{enumerate}
	We may furthermore write down the inner product explicitly:
	\begin{equation}
		\langle f,g\rangle _{\mathcal H(k_h)}=\sum _{k=0}^{\infty}\mathbbm{1}_{\lbrace h^{(k)}(0)>0\rbrace}\frac{1}{h^{(k)}(0)}\frac{f^{(k)}(0)}{\sqrt{k!}}\frac{g^{(k)}(0)}{\sqrt{k!}}.
	\end{equation}
	The space $\mathcal{H}(k_h)$ is therefore a subspace of the space of real analytic functions fulfilling a weighted $\ell ^2$ summability condition with weights dictated by the power series of $h$.
	
	To obtain a description of the RKHS induced by $k=k_h\circ\varphi$, we may use the results of \cite[Theorem 5.7]{paulsen_rkhs} and compute the pullback along $\varphi$ of the space $\mathcal{H}(k_h)$. To do this, let $f_k(x):=\frac{f^{(k)}(0)}{k!}x^k$ for $k\in\mathbb{N}_0$ be the coefficient of the $k$-th term in the Taylor-expansion of a function $f\in\mathcal{H}(k_h)$. Then $(f_k\circ\varphi )(x)=\frac{f^{(k)}(0)}{k!}(ax-c)^k=\frac{f^{(k)}(0)a^k}{k!}(x-c/a)^k$. Therefore, the pullback of the Taylor series induces a shift about the point $c/a$ and a scaling of the weights by a factor $a^k$. Since the Taylor expansion is invariant under the choice of point around which the series is developed, the resulting space $\mathcal{H}(k)$ is given as asserted. Finally, we verify the reproducing property of the kernel with respect to the inner product. Let $f\in\mathcal{H}(k)$. Then $f(x)=\sum _{k=0}^{\infty}\mathbbm{1}_{\lbrace h^{(k)}(0)>0\rbrace}\frac{f^{(k)}(c/a)}{k!}(x-c/a )^k$. The kernel $k(\cdot ,y)$ has the form $k_h(x,y)=\sum _{k=0}^{\infty}\frac{v_k}{k!}(x-c/a )^k$, where $v_k=a^{2k}h^{(k)}(0)(y-c/a)^k$. We therefore obtain
	\begin{equation*}
		\begin{aligned}
			\langle f,k(\cdot, y)\rangle _{\mathcal{H}(k)}&=\sum _{k=0}^{\infty}\mathbbm{1}_{\lbrace h^{(k)}(0)>0\rbrace}\frac{1}{a^{2k}h^{(k)}(0)}\frac{f^{(k)}(\frac{c}{a})}{\sqrt{k!}}\frac{a^{2k}(y-\frac{c}{a})^k}{\sqrt{k!}}\\
								      &=\sum _{k=0}^{\infty}\mathbbm{1}_{\lbrace h^{(k)}(0)>0\rbrace}\frac{f^{(k)}(\frac{a}{c})}{k!}\left( x-\frac{a}{c}\right) ^k=f(y).
		\end{aligned}
	\end{equation*}
To complete the proof, we observe that the norm is precisely the one induced by the inner product.
\end{proof}
\begin{proof}[Proof of \Cref{prop:4_rkhs}]
	Let $h(t):=e^{-\alpha ^2/\beta}p(t)e^{t}$ and define $k$ as in the statement of the Proposition. Then $k(x,y)=h((\sqrt{\beta}x-\alpha /\sqrt{\beta} ) (\sqrt{\beta}y-\alpha /\sqrt{\beta} ))$ and therefore $k$ is of the form specified in \Cref{lem:rkhs_general} with weight sequence $w=(w_k)_{k\geq 0}$ given by $w_k=h^{(k)}(0)$. It is easy to show via iterated use of the product rule that
	\begin{equation}
		h^{(k)}(x)=e^{-\frac{\alpha ^2}{\beta}}(p(\cdot )e^{\cdot})^{(k)}(x)=e^{-\frac{\alpha ^2}{\beta}}\sum _{l=0}^{k\wedge\text{deg}(p)}\binom{k}{l}p^{(l)}(x)e^{x}.
	\end{equation}
	Evaluating at $x=0$ yields the asserted form of the weighs and the assertions on the form of the RKHS induced by $k$ thus follow by \Cref{lem:rkhs_general}. It therefore remains to be shown that $k$ is fully consistent in the sense of \Cref{def:fully consistent_kernel}. For this to hold, we must have that $k(\cdot,y)\in\mathcal{U}$ for all $y\in\R _+$. By the Jordan normal form, any element $g\in\mathcal{U}$ can be written as a sum of products of polynomials with exponentials. Therefore, observe that $k$ must be of the form
	\begin{equation}
		k(x,y)=(q_0(y)+q_1(y)x+...+q_d(y)x^d)e^{\lambda (y)x}.
	\end{equation}
	for some functions $q_0,\dots ,q_d,\lambda :\R _+\rightarrow\R$. Now,
	\begin{equation}
		\begin{aligned}
			&p\left(\left(\sqrt{\beta}x-\frac{\alpha}{\sqrt{\beta}}\right) \left(\sqrt{\beta}y-\frac{\alpha}{\sqrt{\beta}} \right)\right)\\
            &=\sum _{k=0}^da_k\left(\sqrt{\beta}x-\frac{\alpha}{\sqrt{\beta}} \right) ^k\left(\sqrt{\beta}y-\frac{\alpha}{\sqrt{\beta}} \right) ^k\\
																		    &=\sum _{k=0}^da_k\left(\sqrt{\beta}y-\frac{\alpha}{\sqrt{\beta}} \right) ^k\sum _{l=0}^k\binom{k}{l}\sqrt{\beta}^lx^l\left(\frac{\alpha}{\sqrt{\beta}}\right) ^{k-l}\\
																		    &=\sum _{l=0}^dx^l\underbrace{\sum _{k=l}^d\binom{k}{l}\sqrt{\beta}^l\left(\frac{\alpha}{\sqrt{\beta}}\right) ^{k-l}a_k\left(\sqrt{\beta}y-\frac{\alpha}{\sqrt{\beta}} \right) ^k}_{b_l(y)}\\
																		    &=\sum _{l=0}^db_l(y)x^l.
		\end{aligned}
	\end{equation}
	We may set $q_i(y):=e^{-\alpha y}b_i(y)$ and $\lambda (y):=\sqrt{\beta}y-\alpha$ which shows that $k$ is of the desired form and thus indeed fully consistent.
\end{proof}
\begin{proof}[Proof of \Cref{prop:4_fully consistent_ker}]
	Let $k_i$ be defined for $i=1,\dots ,d+1$ as in the Proposition. Then, by the results of \Cref{prop:4_rkhs}, $k_i$ is a fully consistent kernel for $i=1,\dots ,d+1$. Now since $k=k_1+...+k_{d+1}$, we see from \cite{paulsen_rkhs} and \cite{aronszajn} that $k$ is a kernel as it is a sum of kernels. Furthermore, since $\mathcal{U}$ is a vector space and $k_i(\cdot ,y)\in\mathcal{U}$ for $i=1,\dots ,d+1$, $k(\cdot ,y)\in\mathcal{U}$ and therefore $k$ is fully consistent.

	Finally, the rest of the Proposition follows from \cite[Theorem 5.4]{paulsen_rkhs}.
\end{proof}
\begin{proof}[Proof of \Cref{prop:4_existence_min}]
	Let $E^d(k_p\cdot k_{\exp})$ be defined as in the statement of the Theorem. By \Cref{lem:weak_optimisation}, it is sufficient to show that $E^d(k_p\cdot k_{\exp})$ is weakly closed. Without loss of generality, we may assume $\beta = 1,\alpha = 0$ and $d=1$.
	\begin{enumerate}
		\item Let $p(x,x)>0$ for all $x\in\R _+$. This implies $k(x,x)>0$ for all $x\in\R _+$. Let now $g_n:=\eta _nk(\cdot ,y_n)$ with $g_n\rightarrow f$ for some $f\in\mathcal H(k)$. Consider first the case $y_n\rightarrow y_{\infty}$ for some $y_{\infty}\in\R _+$. Then $k(\cdot, y_n)\rightarrow k(\cdot, y_{\infty})$. Since $f\leftarrow g_n=\eta _nk(\cdot ,y_n)$, $\vert\eta _n\vert <\infty$ for all $n$. Therefore, $\eta _n$ has a convergent subsequence. By the uniqueness of the limit, we have 
			\begin{equation}
				f=\lim _{n\rightarrow\infty}g_n=\lim _{n\rightarrow\infty}\eta _nk(\cdot ,y_n)=\eta _{\infty}k(\cdot ,y_{\infty})\in E^1(k).
			\end{equation}
			Assume now $\vert y_n\vert\rightarrow\infty$. Since $g_n\rightarrow f$, we have $\Vert g_n\Vert _{\mathcal H(k)} <C$ for some $C\in\R$. We therefore have
			\begin{equation}
				\Vert g_n\Vert _{\mathcal H(k)}=\Vert\eta _nk_{y_n}\Vert _{\mathcal H(k)}=\vert\eta _n\vert\Vert k_{y_n}\Vert _{\mathcal H(k)}=\vert\eta _n\vert \sqrt{k(y_n,y_n)}<C,
			\end{equation}
			and thus, $\vert\eta _n\vert <\frac{C}{\sqrt{k(y_n,y_n)}}=\frac{Ce^{-y_n^2/2}}{p(y_n,y_n)}$. By \Cref{lem:convergence}, it suffices to consider poitwise convergence. We obtain now, for $n$ large enough
			\begin{equation}
				f(x)\leq 2g_n(x)=2\eta _nk(x,y_n)\leq\frac{Cp(x,y_n)}{p(y_n,y_n)}e^{xy_n-y_n^2/2}\rightarrow 0
			\end{equation}
			and thus $f\equiv 0\in E^1(k)$.
		\item If $p(x,x)=0$ for some $x\in\R _+$, we may use \Cref{lem:q_factorisation2} and obtain polynomials $q:\R _+\rightarrow\R$ and $r:\R _+\times\R _+\rightarrow\R$ with $r(x,x)>0$ for all $x\in\R _+$ and $p(x,y)=q(x)q(y)r(x,y)$ as asserted. By the results of the first statement, we have that $E^1(k_r\cdot k_{\exp})$ is weakly closed. Using \Cref{prop:4_isometry} and \Cref{cor:isometry}, the map $M_q:\mathcal H(r)\cdot\mathcal H(\exp)\rightarrow \mathcal H(p)\cdot\mathcal H(\exp)$ defined as $M_q:f\mapsto qf$ is a bijective isometry and, since bijective isometries map closed sets to closed sets, we have $M_q(E^1(k_p\cdot k_{\exp}))=qE^1(k_p\cdot k_{\exp})\subset\mathcal H(p)\cdot\mathcal H(\exp)$ is weakly closed. This concludes the proof.
	\end{enumerate}
\end{proof}i
\newpage
\printbibliography

@article{manton2015primer,
  title={A primer on reproducing kernel hilbert spaces},
  author={Manton, Jonathan H and Amblard, Pierre-Olivier and others},
  journal={Foundations and Trends{\textregistered} in Signal Processing},
  volume={8},
  number={1--2},
  pages={1--126},
  year={2015},
  publisher={Now Publishers, Inc.}
}

@article{nelson1987parsimonious,
  title={Parsimonious modeling of yield curves},
  author={Nelson, Charles R and Siegel, Andrew F},
  journal={Journal of business},
  pages={473--489},
  year={1987},
  publisher={JSTOR}
}

@article{aronszajn,
 ISSN = {00029947},
 URL = {http://www.jstor.org/stable/1990404},
 author = {N. Aronszajn},
 journal = {Transactions of the American Mathematical Society},
 number = {3},
 pages = {337--404},
 publisher = {American Mathematical Society},
 title = {Theory of Reproducing Kernels},
 urldate = {2023-04-01},
 volume = {68},
 year = {1950}
}

@book{brezis,
	author = {Brezis, Haim},
	year = {2010},
	month = {01},
	pages = {},
	title = {Function Analysis, Sobolev Spaces and Partial Differential Equations},
	publisher = {Springer-Verlag},
	isbn = {978-0-387-70913-0},
	doi = {10.1007/978-0-387-70914-7}
}

@article{carroll,
 ISSN = {00029890, 19300972},
 URL = {http://www.jstor.org/stable/2311361},
 author = {F. W. Carroll},
 journal = {The American Mathematical Monthly},
 number = {1},
 pages = {42--42},
 publisher = {Mathematical Association of America},
 title = {A Polynomial in Each Variable Separately is a Polynomial},
 urldate = {2024-06-07},
 volume = {68},
 year = {1961}
}

@book{folland,
  title={Real Analysis: Modern Techniques and Their Applications},
  author={Folland, G.B.},
  isbn={9781118626399},
  series={Pure and Applied Mathematics: A Wiley Series of Texts, Monographs and Tracts},
  url={https://books.google.at/books?id=wI4fAwAAQBAJ},
  year={2013},
  publisher={Wiley}
}

@book{jacod,
	title={Limit Theorems for Stochastic Processes},
	author={Jacod, J. and Shiryaev, A.N.},
	isbn={9783540178828},
	lccn={lc87009865},
	series={Grundlehren der mathematischen Wissenschaften},
	url={https://books.google.at/books?id=sUgXKpUIdHwC},
	year={1987},
	publisher={Springer Berlin Heidelberg}
}

@Inbook{bjoerk,
author="Bj{\"o}rk, Tomas",
editor="Carmona, Ren{\'e} A.
and {\c{C}}inlar, Erhan
and Ekeland, Ivar
and Jouini, Elyes
and Scheinkman, Jos{\'e} A.
and Touzi, Nizar",
title="On the Geometry of Interest Rate Models",
bookTitle="Paris-Princeton Lectures on Mathematical Finance 2003",
year="2004",
publisher="Springer Berlin Heidelberg",
address="Berlin, Heidelberg",
pages="133--215",
isbn="978-3-540-44468-8",
doi="10.1007/978-3-540-44468-8_2",
url="https://doi.org/10.1007/978-3-540-44468-8_2"
}

@article{cuchiero_naflvr,
author = {Cuchiero, C. and Klein, I. and Teichmann, J.},
title = {A New Perspective on the Fundamental Theorem of Asset Pricing for Large Financial Markets},
journal = {Theory of Probability \& Its Applications},
volume = {60},
number = {4},
pages = {561-579},
year = {2016},
doi = {10.1137/S0040585X97T987879},

URL = { 
    
        https://doi.org/10.1137/S0040585X97T987879
    
    

},
eprint = { 
    
        https://doi.org/10.1137/S0040585X97T987879
    
    

}
}

@PHDTHESIS{filipovic_phd,
	copyright = {In Copyright - Non-Commercial Use Permitted},
	year = {2000},
	author = {Filipović, Damir},
	size = {123 S.},
	language = {en},
	address = {Zürich},
	publisher = {ETH Zürich},
	DOI = {10.3929/ethz-a-003882242},
	title = {{Consistency problems for {HJM} interest rate models}},
	Note = {Diss. Mathematische Wissenschaften ETH Zürich, Nr. 13603, 2000.},
	school = {ETH Zurich}
}

@article{filipovic_invariant,
  title={Invariant manifolds for weak solutions to stochastic equations},
  author={Damir Filipovi{\'c}},
  journal={Probability Theory and Related Fields},
  year={2000},
  volume={118},
  pages={323-341},
  url={https://api.semanticscholar.org/CorpusID:6968887}
}

@Article{filipovic_kr,
author={Camenzind, Nicolas
and Filipovi{\'{c}}, Damir},
title={Stripping the Swiss discount curve using kernel ridge regression},
journal={European Actuarial Journal},
year={2024},
month={Aug},
day={01},
volume={14},
number={2},
pages={371-410},
issn={2190-9741},
doi={10.1007/s13385-024-00386-4},
url={https://doi.org/10.1007/s13385-024-00386-4}
}

@article{tappe_invariant,
author = {Damir Filipović and Stefan Tappe and Josef Teichmann},
title = {{Invariant manifolds with boundary for jump-diffusions}},
volume = {19},
journal = {Electronic Journal of Probability},
number = {none},
publisher = {Institute of Mathematical Statistics and Bernoulli Society},
pages = {1 -- 28},
year = {2014},
doi = {10.1214/EJP.v19-2882},
URL = {https://doi.org/10.1214/EJP.v19-2882}
}

@article{jarrow_wu,
author = {Wu, David and Jarrow, Robert},
title = {Fitting Dynamically Consistent Forward Rate Curves: Algorithm and Comparison},
journal = {International Journal of Theoretical and Applied Finance},
volume = {},
number = {},
pages = {},
year = {2024},
doi = {10.1142/S0219024924500213},

URL = {    
        https://doi.org/10.1142/S0219024924500213
},
eprint = {     
        https://doi.org/10.1142/S0219024924500213
}
}

@article{autoencoder,
author = {Lyashenko, Andrei and Mercurio, Fabio and Sokol, Alexander},
title = {Machine Learning for Interest Rates: Using Auto-Encoders for the Risk-Neutral Modeling of Yield Curves},
month = {September},
day = {25},
year = {2024},
URL = {https://ssrn.com/abstract=4967989 or http://dx.doi.org/10.2139/ssrn.4967989}
}

@book{filipovic,
  title={{Term-Structure Models: A Graduate Course}},
  author={Filipovi\'c, D.},
  isbn={9783540680154},
  lccn={2009933038},
  series={Springer Finance},
  url={https://books.google.at/books?id=KqcSh6CavaAC},
  year={2009},
  publisher={Springer Berlin Heidelberg}
}

@article{filipovic3,
title = {Stripping the Discount Curve - a Robust Machine Learning Approach},
author = {Filipović, Damir and Pelger, Markus and Ye, Ye},
year = {2022},
institution = {Swiss Finance Institute},
type = {Swiss Finance Institute Research Paper Series},
number = {22-24},
url = {https://EconPapers.repec.org/RePEc:chf:rpseri:rp2224}
}

@article{filipovic_discount,
author={Filipovi{\'{c}}, Damir},
title={Discount models},
journal={Finance and Stochastics},
year={2023},
month={Oct},
day={01},
volume={27},
number={4},
pages={933-946},
issn={1432-1122},
doi={10.1007/s00780-023-00514-0},
url={https://doi.org/10.1007/s00780-023-00514-0}
}

@article{filipovic_lrtsm,
author = {Filipovi{\'c}, Damir and Larsson, Martin and Trolle, Anders B.},
title = {Linear-Rational Term Structure Models},
journal = {The Journal of Finance},
volume = {72},
number = {2},
pages = {655-704},
doi = {https://doi.org/10.1111/jofi.12488},
url = {https://onlinelibrary.wiley.com/doi/abs/10.1111/jofi.12488},
eprint = {https://onlinelibrary.wiley.com/doi/pdf/10.1111/jofi.12488},
year = {2017}
}

@Article{hjm,
  author={Heath, David and Jarrow, Robert and Morton, Andrew},
  title={{Bond Pricing and the Term Structure of Interest Rates: A New Methodology for Contingent Claims Valuation}},
  journal={Econometrica},
  year=1992,
  volume={60},
  number={1},
  pages={77-105},
  month={January},
  doi={},
  url={https://ideas.repec.org/a/ecm/emetrp/v60y1992i1p77-105.html}
}

@article{musiela,
	author={Musiela, M.},
	title={{Stochastic PDEs and term structure models}},
	journal={Journees Internationales de Finance},
	publisher={IGR-AFFI},
	year={1993},
}

@article{teichmann1,
author = {Filipović, Damir and Teichmann, Josef},
year = {2003},
month = {02},
pages = {398-432},
title = {{Existence of Invariant Manifolds for Stochastic Equations in Infinite Dimension}},
volume = {197},
journal = {Journal of Functional Analysis},
doi = {10.1016/S0022-1236(03)00008-9}
}

@book{paulsen_rkhs, 
	place={Cambridge}, 
	series={Cambridge Studies in Advanced Mathematics}, 
	title={An Introduction to the Theory of Reproducing Kernel Hilbert Spaces}, 
	DOI={10.1017/CBO9781316219232}, 
	publisher={Cambridge University Press}, 
	author={Paulsen, Vern I. and Raghupathi, Mrinal}, year={2016}, 
	collection={Cambridge Studies in Advanced Mathematics}
}

@book{musiela_rutkowski,
  title={Martingale Methods in Financial Modelling},
  author={Musiela, M. and Rutkowski, M.},
  isbn={9783540614777},
  lccn={97011586},
  series={Applications of Mathematics : Stochastic Modelling and Applied Probability},
  url={https://books.google.at/books?id=_dfHAAAAIAAJ},
  year={1997},
  publisher={Springer}
}

@book{zabczyk, 
	place={Cambridge}, 
	series={Encyclopedia of Mathematics and its Applications}, 
	title={Stochastic Equations in Infinite Dimensions}, 
	publisher={Cambridge University Press}, 
	author={Da Prato, Guiseppe and Zabczyk, Jerzy}, 
	year={1992}, 
	collection={Encyclopedia of Mathematics and its Applications}}

@book{EK,
  title={{Markov Processes: Characterization and Convergence}},
  author={Ethier, Stewart N. and Kurtz, Thomas G.},
  isbn={9780470317327},
  lccn={85012078},
  series={Wiley Series in Probability and Statistics},
  url={https://books.google.at/books?id=zvE9RFouKoMC},
  year={2009},
  publisher={Wiley}
}

@book{filipovic_consistency, 
    series={Lecture Notes in Mathematics; 1760}, 
    title={{Consistency Problems for Heath-Jarrow-Morton Interest Rate Models}}, 
    url={https://infoscience.epfl.ch/handle/20.500.14299/49680}, 
    DOI={10.1007/b76888}, 
    abstractNote={Research monograph providing appropriate consistency conditions for and examples of blended models for the term structure of interest rates within the Health-Jarrow-Morton framework, combining curve-fitting methods and factor models. Softcover.}, 
    publisher={Springer}, 
    author={Filipović, Damir}, 
    year={2001}, 
    collection={Lecture Notes in Mathematics; 1760} }
\end{document}